\documentclass[12pt]{article}

\usepackage[utf8]{inputenc}
\usepackage{microtype}
\usepackage[english]{babel}
\usepackage{amssymb}  
\usepackage{ifsym}  

\usepackage{amsthm} 

\usepackage{amsfonts} 
\usepackage{mathtools} 
\usepackage{dsfont} 
\usepackage{forest}
\usepackage{tikz} 
\usetikzlibrary{shapes,arrows,positioning,shadows,trees}
\usepackage{algorithm} 
\usepackage{algpseudocodex} 
\usepackage[abbreviations, automake, nogroupskip, nonumberlist]{glossaries-extra} 
\usepackage[style=ieee, maxbibnames=5, backend=biber]{biblatex}
\usepackage{hyperref} 

\usepackage[capitalise, noabbrev]{cleveref} 
\usepackage{siunitx} 
\usepackage{tabularray} 

\usepackage[autostyle=false, style=english]{csquotes} 
\MakeOuterQuote{"} 
\usepackage{graphicx}
\usepackage{array, makecell} 
\usepackage{booktabs} 
\usepackage{subcaption} 

\usepackage{rotating}
\input{acronyms.tex}
\theoremstyle{plain}
\newtheorem{theorem}{Theorem}
\newtheorem{proposition}[theorem]{Proposition}

\newcounter{TodoCounter}

\DeclareMathOperator*{\argmax}{arg\,max}

\renewcommand{\vec}[1]{\boldsymbol{#1}}         

\DeclarePairedDelimiter\floor{\lfloor}{\rfloor} 

\DeclareMathOperator{\EX}{\mathbb{E}}

\newcommand{\N}{\mathbb{N}}
\newcommand{\R}{\mathbb{R}}
\newcommand{\e}{\mathrm{e}}

\newcommand{\depthLimit}{d_{\max}}
\newcommand{\iterLimit}{n_{\textnormal{iter}}}
\newcommand{\maxReachedDepth}{d}
\newcommand{\currentState}{\text{state}}

\newcommand{\et}[1]{\varepsilon_{#1}}

\AtBeginDocument{\DeclareSIUnit{\kWh}{kWh}}

\let\oldfrac\frac
\renewcommand{\frac}[2]{%
  \mathchoice
    {\oldfrac{#1}{#2}}
    {#1/#2}
    {#1/#2}
    {#1/#2}
}

\begin{document}
\title{
Online Dynamic Pricing for Electric Vehicle Charging Stations with Reservations
}

\author{
        Jan Mrkos\thanks{Faculty of Electrical Engineering
        Czech Technical University in Prague, 
        Karlovo náměstí 13, Prague 2, 121 35} \thanks{Corresponding author, \tt\small mrkosja1@fel.cvut.cz}, Antonín Komenda\footnotemark[1], David Fiedler\footnotemark[1], and Jiří Vokřínek\footnotemark[1]
}
\maketitle

\begin{abstract}
        This paper introduces a novel model for online dynamic pricing of electric vehicle charging services that integrates reservation, parking, and charging into a comprehensive bundle priced as a whole.
        Our approach focuses on the individual high-demand, fast-charging location, employing a Poisson process as a model of charging reservation arrivals, and develops an online dynamic pricing strategy optimized through a Markov Decision Process (MDP).
        
        A key contribution is the novel analysis of discretization error introduced when incorporating the continuous-time Poisson process into the discrete MDP framework.
        The MDP model's feasibility is demonstrated with a heuristic dynamic pricing method based on Monte-Carlo tree search, offering a viable path for real-world applications.
\end{abstract}


\section{Introduction}
\label{sec:introduction}
The mass electrification of personal transportation, concurrent with an increasing proportion of energy sources being renewable and thus, intermittent, will significantly impact the electricity grid. 
Most likely, these changes will mean that drivers will have to adapt to a new way of using their soon-to-be electric vehicles (EVs). 
Charging an EV takes more time than filling a tank with gas, and the construction of new EV charging stations (CSs) is subject to restrictions from the power grid. 
As a result, in some situations, such as in fast CSs along highways, demand may outweigh supply. 


To make the allocation of charging capacity more efficient, it is reasonable to dynamically price the charging of EVs. 
In this way, dynamic pricing can help CS operators maximize revenue and reduce congestion by responding to changing demand for charging.
In addition, this way, price signals from the energy markets and grid operators can be propagated to EV users.

However, when traveling long distances, charging an EV is already more complicated than refueling a gasoline car and requires a lot more planning. 
Adding price uncertainty to the mix could make planning and travel even more difficult. 
However, if EV charging could be reserved in advance, then the price could be fixed at the time of reservation. 
This way, the driver could plan the whole trip and know the exact cost of charging beforehand. 

Although reservations reduce uncertainty for both the EV drivers and the CS operators, they also introduce a new layer of complexity and reduce flexibility for the EV drivers.
However, this complexity can be reduced by using driver assistant technologies, such as smartphone apps, to make reservations on behalf of drivers.
For example, EV routing algorithms can plan entire trips, including charging stops, and let the driver select the best trade-off between trip time and cost~\cite{cuchyMultiObjectiveElectricVehicle2024a}. 
Reservations are a natural extension of this approach.

In this paper, we present a model of dynamic pricing of EV charging.
It relies on the full reservation of the charging service, which is priced as a single bundle including the energy delivered to the EV, reservation cost, time spent at the CS, and the parking spot.
We consider fully dynamic pricing, i.e., the price of each requested reservation can be different, reflecting the current state of the CS, the parameters of the request, and the expected future demand for charging.
Finally, we focus on revenue maximization for the seller, operator of a single charging location. 
This point of view has become more popular in recent years~\cite{fangDynamicPricingElectric2021,dahiwaleComprehensiveReviewSmart2024} as many realized that the transition to electric mobility and its sustainability is not feasible without profitable EV CSs. 
Arguably, focusing on a single location is limiting, and many works explore the problem on a much larger scale. 
However, limiting our focus makes the complexity raised by the first two points manageable.

\subsection{Literature review}

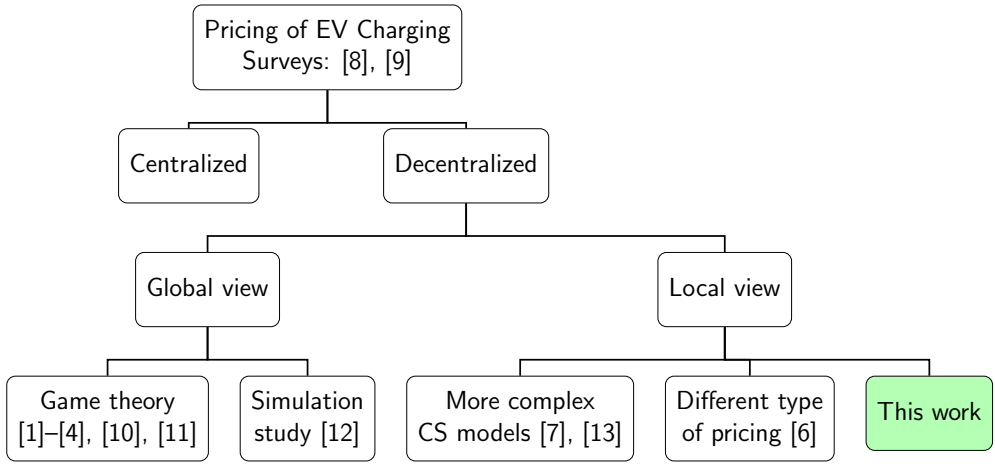
\begin{figure}
  \begin{center}
  \resizebox{0.99\textwidth}{!}{
    \tikzset{parent/.style={align=center,text width=3cm,rounded corners=3pt},
child/.style={align=center,text width=3cm,rounded corners=3pt}
}

\begin{forest}
for tree={
    rectangle, rounded corners, draw, align=center,
    text centered, 
    inner sep=5pt, 
    minimum height=1.2cm, 
    font=\sffamily, 
    edge={thick, draw}, 
    l=2cm, 
    s sep=2cm, 
    anchor=north,
    child anchor=north,
    parent anchor=south,
    edge path={
                    \noexpand\path[\forestoption{edge}]
                    (!u.parent anchor) -| +(0,-15pt) -| (.child anchor)\forestoption{edge label};
                },
  }
[Pricing of EV Charging\\ Surveys: \cite{rigasManagingElectricVehicles2015a,dahiwaleComprehensiveReviewSmart2024}
    [Centralized]
    [Decentralized, s sep=0.5cm
        [Global view, s sep=0.5cm
          [Game theory\\ \cite{yuHierarchicalGameNetworked2021,gaoPriceBasedIterativeDouble2022,houSimultaneousMultiRoundAuction2022,zhangRangeConstrainedStableFederated2024,houBiddingPreferredTiming2020, boatengAutomatedValetParking2024}]
          [Simulation\\ study \cite{seyedyazdiCombinedDriverStationInteractive2020}]
        ]
        [Local view, s sep=0.5cm 
          [More complex\\ CS models \cite{zhangOptimalChargingScheduling2019,abdalrahmanDynamicPricingDifferentiated2022}]
          [Different type\\ of pricing \cite{fangDynamicPricingElectric2021}]
          [This work, fill=green!30]
        ]
    ]
]
\end{forest}
  }
  \caption{
    Structure of the literature survey. 
    }
  \label{fig:sota_tree}
  \end{center}
\end{figure}

Pricing of EV charging is an active research topic, and many papers tackle the problem from different points of view. We provide an overview of the structure of our literature survey in \cref{fig:sota_tree}.

The works on the pricing of EV charging are commonly divided between centralized and decentralized approaches, as is also the case in the surveys \cite{rigasManagingElectricVehicles2015a,dahiwaleComprehensiveReviewSmart2024}. 
A centralized approach collects all the inputs for pricing with a central authority that then makes and distributes decisions. 
However, the decentralized approach allows participants to make decisions independently. 
As centralization is practically problematic in the context of public EV charging, many works, including ours, focus on decentralized approaches. 

One branch of decentralized research takes a ``global'' view of the EV charging ecosystem, modeling both EV drivers and CSs individually in their interactions in the transportation network.
The second branch of research, which includes this work as well, focuses on the ``local'' viewpoint of a single actor in the ecosystem, such as the CS, without the ambition of optimizing some ecosystem-wide criteria.

The first segment of papers that take the ``global'' view uses a game theoretical approach to the problem of EV charging. 
In~\cite{yuHierarchicalGameNetworked2021,gaoPriceBasedIterativeDouble2022,houSimultaneousMultiRoundAuction2022,zhangRangeConstrainedStableFederated2024,houBiddingPreferredTiming2020}, the allocation arises as part of iterative negotiation between the CSs and the EV drivers. Price changes are used as signals to influence the decisions of the drivers and to guarantee some useful properties, such as individual rationality~\cite{gaoPriceBasedIterativeDouble2022,zhangRangeConstrainedStableFederated2024} or incentive compatibility and preservation of privacy~\cite{zhangRangeConstrainedStableFederated2024}.
While these approaches do not rely on a central authority to make decisions, they require a widespread consensus on the rules and protocols of the system.
An additional drawback of this approach is that it requires many iterations of negotiation between drivers and the CSs. 
This can be addressed by having agents in the form of, e.g., smartphone applications perform the negotiation on drivers' behalf. 
Overall, these approaches are often not practically scalable, as the required consensus between CS owners and EV drivers is difficult to achieve.

Other works that take the ``global'' view of electromobility look at the EV drivers and CSs individually without applying game theoretical concepts. For example, \textcite{seyedyazdiCombinedDriverStationInteractive2020} proposes a selection algorithm for EV drivers to choose the best CS based on the trip price travel time, as well as a pricing algorithm that considers previously reserved pricing requests and day-ahead pricing of electricity. The proposed algorithms are then allowed to interact in a simulation. Simulation results show that the proposed solution improves upon the selected baselines. However, the work suffers from the same coordinational weaknesses as the game-theoretical approaches; it assumes all the EV drivers and CSs are all using the proposed algorithms. 

\subsubsection{Local view of pricing}

We follow the ``local'' branch of research, focusing on an individual CS or CS operator and their pricing strategy. 
Here, the EV drivers' needs are abstracted into a much simpler model, usually a stochastic demand model. This approach is useful when the CSs are modeled in greater detail or when the interactions are more complex. Representatives of this approach are, for example, \textcite{zhangOptimalChargingScheduling2019} and \textcite{abdalrahmanDynamicPricingDifferentiated2022}, who propose a model of pricing EV charging that considers different charging speeds and prices them differently. 

Within the local focus on the CSs point of view, we are interested in reservations and the form that the dynamic pricing of these reservations has. 

\subsubsection{Reservations}
Reservations are not common in the EV charging industry.
However, for example, CS operator EVgo\footnote{See https://www.evgo.com/reservations/} has recently introduced a reservation system for their fast CSs, with the company claiming an interest in the service by customers traveling long-distance trips~\cite{edelsteinEVFastchargingReservations2021}. These reservations are offered for a flat fee and can be made up to 24 hours in advance. 
Reservation systems are useful for CS operators as they can better plan the utilization of their resources.  For EV drivers, reservations reduce uncertainty about their trips. 

In the ``local'' literature on EV charging, reservations are not particularly common either. 
They are present implicitly or explicitly in the auction-like mechanisms that result in the allocation of the charging capacity, such as \cite{yuHierarchicalGameNetworked2021,gaoPriceBasedIterativeDouble2022,boatengAutomatedValetParking2024}.
Other works assume a less organized EV ecosystem and model reservations explicitly, such as \cite{fangDynamicPricingElectric2021,kumarChaseMeHeuristicScheme2022}. 
\textcite{fangDynamicPricingElectric2021} is of most interest to us as it studies the pricing of a single CS that offers reservations, and the pricing responds to previously reserved demand. 
 
\subsubsection{Dynamic pricing}
Dynamic pricing is a relatively common feature of papers on EV charging, as illustrated in the survey~\cite{rigasManagingElectricVehicles2015a}. 
However, there are big differences between the dynamic pricing schemes. First, different works price different products. Second, dynamic pricing can mean different things in different papers. 
Some works apply dynamic pricing to the electricity price~\cite{zhangOptimalChargingScheduling2019}. 
In that case, the common choice of pricing is extending the smart-grid concept of time-of-use pricing to EV charging~\cite{wuOnlineEVCharge2022}, which works by setting different prices per \unit{\kWh} for different time intervals, usually well ahead of time.
Other works apply dynamic pricing to the per-minute fee, which is viewed as an improvement over the current state of charging per \unit{\kWh} as it incentivizes drivers to disconnect when their vehicle is sufficiently charged~\cite{bakhtiariComparativeAnalysisTimeBased2024}.
Others don't look at the charging and instead focus on the pricing of parking reservations, e.g., ~\textcite{leiDynamicPricingReservation2017a,fangDynamicPricingElectric2021}. 

\subsubsection{Other related works}
\label{sec:other_related_work}
To apply the dynamic pricing method described in this paper, CS operators would first need to instantiate the demand models that the method uses. 
In this work, we use a Poisson process to model the charging request arrivals and simple parametric distributions to model the charging durations and request lead times (time between the request and the charging session). 
However, the pricing method does not rely on any specific form of a demand model, with the exception of the Poisson process for the charging requests arrivals, and we allow for non-homogenous (i.e. time-dependent) demand. 
The literature on forecasting the demand for EV charging is rich; for example, see survey by~\textcite{rashidComprehensiveSurveyElectric2024}. 
The most relevant to our work are then articles that focus on short- to medium-term single CS demand forecasting~\cite{orzechowskiDatadrivenFrameworkMediumterm2023}. 
Although we use parametric distributions in our experiments for the sake of simplicity, our heuristic algorithm can also accommodate other demand models, be it sampling-based probabilistic techniques~\cite{suForecastElectricVehicle2017} or black-box models based on neural networks~\cite{zhuElectricVehicleCharging2019}. 


The currently deployed pricing solutions for fast CSs vary greatly. 
\textcite{chaturvediGenerativeAIApproach2023} catalogs them and notes the fragmented and opaque nature of the pricing strategies. 
Most charging schemes are either energy-based (e.g.: \SI{0.35}[\$]{\per\kWh}), time-based (e.g.: \SI{2.00}[\$]{\per\hour}), or a combination of different rates for different durations or energy consumption (e.g.: \SI{0.80}[\$]{\per\hour} for first 4 hours, \SI{5.00}[\$]{\per\hour} after that). 
However, free charging is the most common across all CSs but is likely not sustainable in the long term. 
Flatrate, either per hour or \unit{kWh}, is the second most common.

In practice, there are many fees that are associated with charging but are charged separately, such as the \emph{connection fee} for starting the charging sessions, \emph{parking fee} for time spent at the charging location, and \emph{idle fee} for time spent connected to the charger after the charging session is finished, which is meant to incentivize drivers to disconnect after their session is finished. 
These fees vary and are sometimes voided if, for example, the driver charges a sufficient amount of charge. 
Our model proposes a significant simplification over these paid pricing schemes, by bundling all the fees into a single price for the charging session known ahead of time.

\subsection{Contributions}
In this work, we propose an EV charging system with full reservations (following the taxonomy proposed in \textcite{basmadjianInteroperableReservationSystem2019}), and we price these reservations as a complete product, including the reservation, delivered electricity, and the parking spot. 
Our approach is most similar to \textcite{boatengAutomatedValetParking2024}, who let drivers reserve CSs before their arrival at the charging location and propose a dynamic pricing scheme to manage demand.
However, \textcite{boatengAutomatedValetParking2024} focuses on autonomous vehicles and takes the global view of the problem, with both customers and service provider constantly adjusting their strategies to maximize their individual utilities. 

When focusing on pricing, our work is most similar to ~\textcite{fangDynamicPricingElectric2021}, who take the local approach to pricing. 
Their pricing is dynamic in the sense that the price changes in 1-hour timesteps based on demand. 
The price updates are broadcast to all potential customers. 
In contrast, our protocol has the customers request charging reservations, and the station responds with a price.

This paper is a substantial extension of our previous conference paper~\cite{mrkosRevenueMaximizationElectric2018} that experiments with using much simpler \MDP/ models for dynamic pricing of EV charging, and significant rework and extension of journal publication~\cite{mrkosDynamicPricingCharging2022} with similar focus, that lays out basic parts of the \MDP/ model and parts of the solution methods used here.

The main contributions of this work are:
\begin{enumerate}
  \item We developed a new \MDP/-based model for dynamic pricing of EV charging reservations for one charging location (\cref{sec:mdp_model}). 
  Unlike previous works~\cite{zhangOptimalChargingScheduling2019,wuOnlineEVCharge2022,abdalrahmanDynamicPricingDifferentiated2022,bakhtiariComparativeAnalysisTimeBased2024,leiDynamicPricingReservation2017a,fangDynamicPricingElectric2021}, our model uses fully dynamic pricing of the complete \emph{reserved EV charging service}, i.e., all parts of the service (charging, reservation, parking, etc.) are priced together, and the final price, based on the expected demand and the price of electricity, can be different for each reservation request.
  \item Since we use discrete \MDP/ with charging request arrivals modeled by the continuous-time Poisson process and thus inevitably introduce a discretization error, we define a novel metric that quantifies this approximation error (\cref{sec:transition_poisson}).
  \item To validate the feasibility of our approach, we developed a scalable heuristic solution method for the model based on Monte-Carlo tree search for pricing of the EV charging service reservations (\cref{sec:msts_pricing_algorithm}), and compared its performance to a number of baselines in \cref{sec:experimental_results}.
\end{enumerate}



\section{MDP model for dynamic pricing of reserved EV charging}
\label{sec:mdp_model}
In our work, the pricing of EV charging is a question of sequential pricing decisions for incoming charging reservation requests. As such, \MDP/ is a natural model for the problem. \MDP/ is the model of choice for other dynamic pricing works, especially when combined with reinforcement learning \cite{abdalrahmanDynamicPricingDifferentiated2022,seyedyazdiCombinedDriverStationInteractive2020,zhangDeepRLBasedAlgorithm2022}. 
In these works, the state space is constructed with as many state variables as possible, and the pricing policy is learned from the data.
However, we avoid the use of reinforcement learning in our work due to the large computing requirements to train the model and the currently low availability of representative data to train the model on. Instead, we focus on a concise model of the problem and a scalable heuristic solution method.

Lastly, the demand model in our work uses a Poisson process to model the arrival of the charging requests. This is a relatively common choice in the literature on dynamic pricing of EV charging, used by \cite{zhangOptimalChargingScheduling2019,seyedyazdiCombinedDriverStationInteractive2020,fangDynamicPricingElectric2021}. By applying the Poisson process to our discrete-time model of online dynamic pricing, we introduce a discretization error. This error is described and quantified in \cref{sec:transition_poisson}.

\begin{figure*}
  \centering
  \def\svgwidth{0.99\textwidth}
  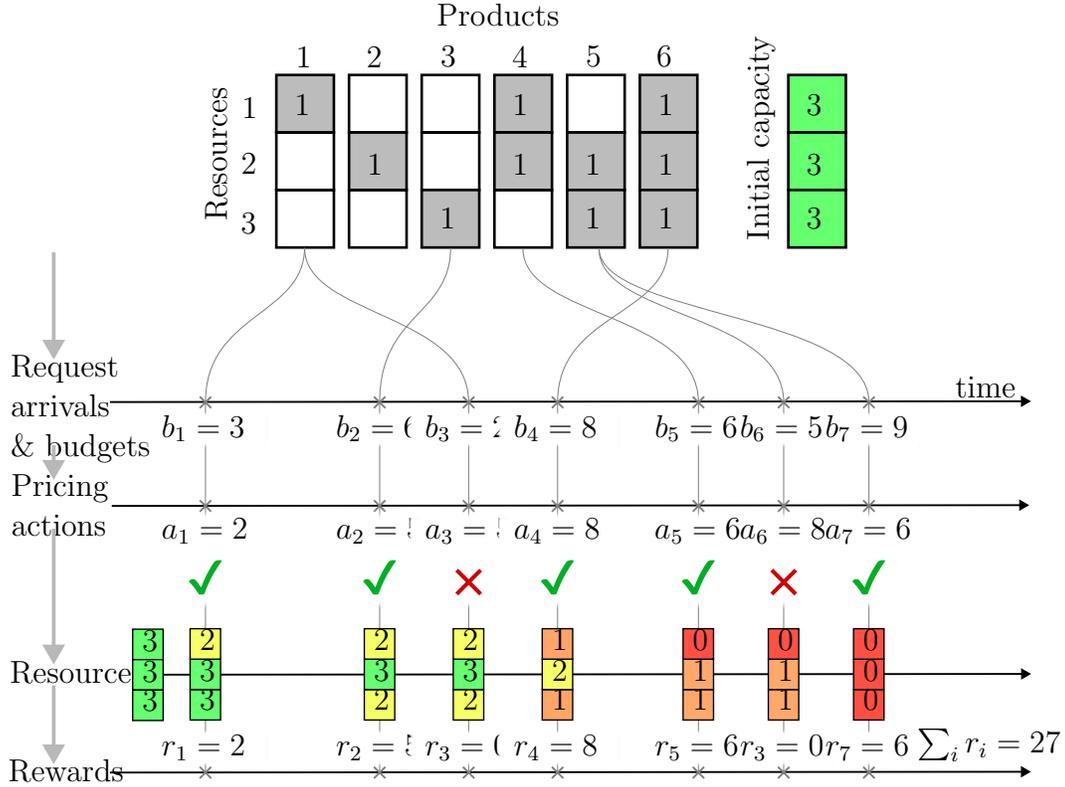
  \caption{
    Illustration of the online dynamic pricing of \EV/ charging. The charging products being priced are vertical columns in the resource matrix on top. For example, product \([1,1,0]\) represents consecutive charging in the first time slot (resource) and second time slot out of the three. In the \EV/ charging problem, we assume each product contains every resource (the charging slot capacity in a time interval) only once and that products contain only consecutive time intervals (i.e., there is no product \([1,0,1]\)). Below the matrix, the \emph{request arrivals and budgets} show customers requesting different products at different times and their budgets (\(\bdg_1\) to \(\bdg_7\)). The next line below shows \emph{pricing actions} \(a_1\) to \(a_7\) taken by the seller. When a customer accepts the price (i.e., when \(\bdg_i \ge a_i\), shown as green \checkmark), the seller accumulates a reward \(r_i\), and his resource capacity is reduced by the product requirements, this is shown on the two bottom lines.
  } 
  \label{dynamic_pricing_model}
\end{figure*}

\subsection{Assumptions and Real-world Application}
The main application of our pricing model is fast CSs along highways. 
We assume these CSs to have a limited number of chargers and a high demand for charging services. 
Additionally, we assume that drivers depart the CS shortly after the charging session is finished. 
This assumption is reasonable for highway CSs, but might not hold for destination chargers, without additional incentives.  
We also assume that drivers request an appropriate number of charging timeslots for their needs, that they arrive on time or shortly after the reservation, and that they leave shortly before the reservation ends.
We do not address the issue of cancellations or modifications of the reservations, not of strategic driver behavior (e.g., drivers reserving multiple charging sessions without intending to use all of them). 

On the electricity supply side, in our experiments, we assume that the charging capacity of the charging location is only limited by the number of available charging slots, and we do not model the power curves of the ongoing charging sessions against changing power availability from the grid or changes in the price of electricity. 
However, the model is extensible and changes in power availability and electricity prices can be incorporated in a straightforward fashion. 
The specific implementation would depend on the type of the specifics of the pricing scheme used to price electricity and the protocol for power availability adjustments. 
Since neither of these are standardized, we leave this extension of these features for future work.

\subsection{Formalization}
In this work, dynamic pricing of EV charging involves setting optimal or near-optimal prices online for different charging products as customer requests for these products arrive, one after another. The seller, the CS operator, combines his resources, the charging capacity of a set of CSs, into products offered to customers, such as half-day charging of an \EV/. 
In our case, the product is a charging session, and the different resources it consists of are charging capacities in different time slots. 

Although the demand for these products is unknown beforehand, the seller has historical data on customer request arrivals and a model of customer responses to changing prices. 
The seller's goal is to price each arriving product request in a way that maximizes its objective function, such as revenue, considering demand uncertainty. 
This optimization occurs over a finite time horizon, after which the resources can no longer be sold. 

Here, we give a minimalist formal description of the pricing problem by considering first the supply side that puts constraints on the seller and the demand side that prompts the seller's actions. The description is illustrated by \cref{dynamic_pricing_model}.

The supply side of the problem is made up of a set of \(n\) resources that can be combined into \(m\) products available for sale. Each product is represented by a vector \(\prd\in\N^n_0\), elements of which prescribe the number of individual resources used in the product. The availability of these products is constrained by the initial capacity of the resources \(\cpc_0 \in \N^n \) and the lengths of the selling periods of different resources \(\vec{T} \in \R^+\) that determine the time after which each resource and product of which it is a part can no longer be sold.  

The demand side of the problem is modeled by a non-homogeneous, compound Poisson counting process \(N(\tau)\), \(\tau\in(0, \max(\vec{T}))\) that models the arrivals of requests for different products in time and distributions of finite internal customer valuations for different products, \(\{\BDG_{\prd} | \prd \in \PRD \}\). Customers accept the price offered for the requested product if it is below their internal valuation. Otherwise, they reject the offer and leave the system.

Realized demand takes the form of a sequence of timestamped product requests (pair \((\prd, \tau)\)) associated with hidden customer valuations of products, \(\bdg_i \sim \BDG_{\prd_i}\): 
\begin{equation}
  \label{eq:request_sequence}
    \CST = ((\prd_1, \tau_1, \bdg_1), (\prd_2, \tau_2, \bdg_2), \ldots )
\end{equation}
The service provider needs to process these requests individually as they arrive. The protocol for a single interaction between the seller and customers is following:
\begin{enumerate}
  \item An \(i\)th customer requests product \(\prd_i\) from the seller at time \(\tau_i\). 
  \item If the request is feasible, that is, if \(\tau_i \le T_{\prd_i}\) and the resource capacity is greater than the products resource requirements, \(\cpc \ge \prd_i\) in all components, then the seller proposes price \(a \in A\subset\R\) to the customer. Otherwise, the seller rejects the request (equivalent to proposing an infinite price for the product).
  \item The customer compares the price \(a\) against their internal valuation of the product \(\bdg_i\), and if \(a \le \bdg_i\), then the customer pays the price and buys the product. Otherwise, the customer rejects the price \(a\) and leaves the system. 
\end{enumerate}
Crucially, the seller \emph{does not} know the values of the customer's valuations \(\bdg_i\). Each customer requests a single product and leaves the system after accepting or rejecting the offered price. 

\subsection{MDP Model}

\begin{figure}
    \centering
    \includegraphics[width=0.99\textwidth]{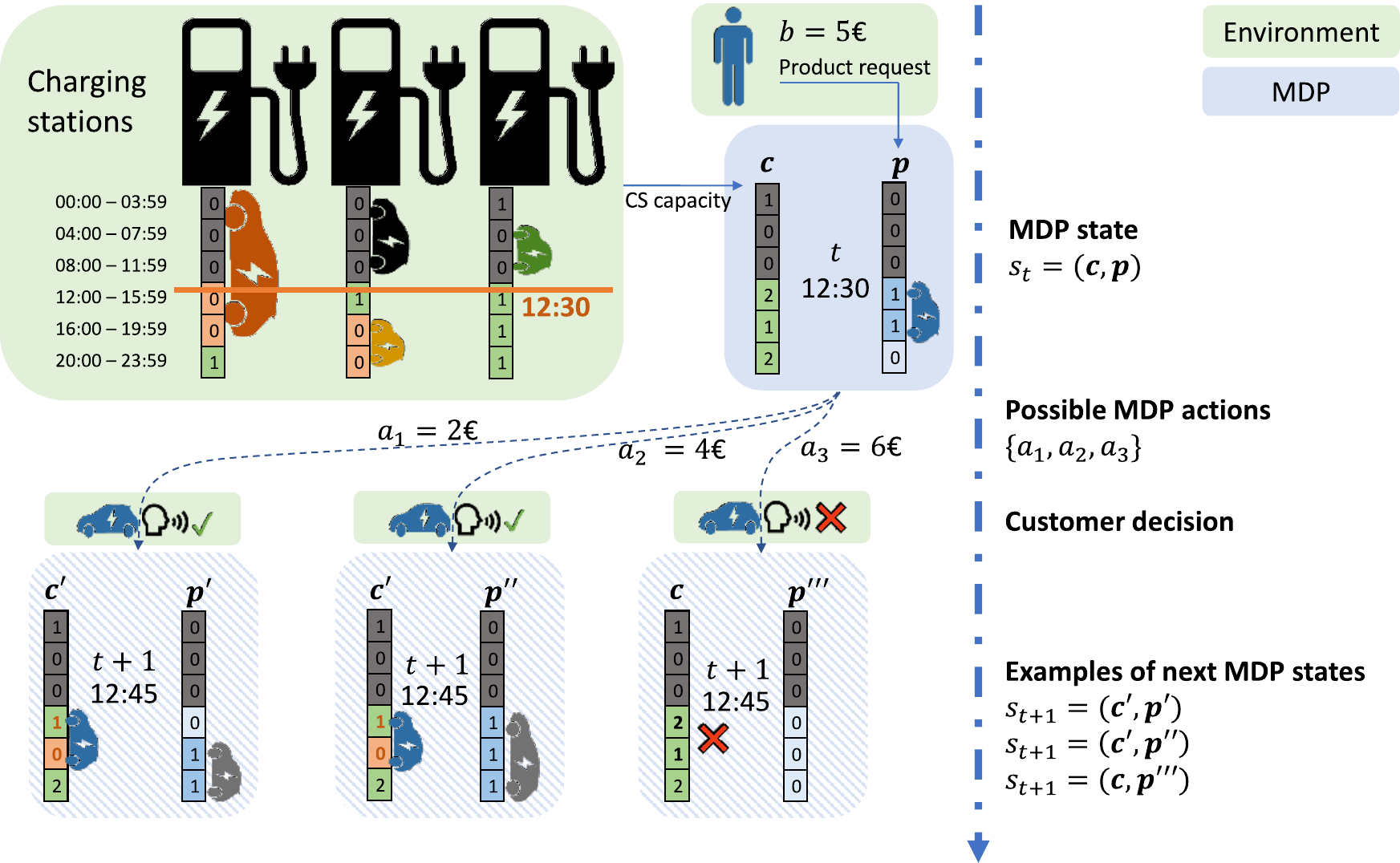}
    \caption{
    Illustration of the \MDP/ states for the dynamic pricing of \EV/ charging. Unlike \cref{dynamic_pricing_model}, this figure illustrates the expiration of resources after their selling period ends (grey in the capacity and product vectors). The blue squares represent the \MDP/ states. At timestep \(t\), the capacity of the CS is expressed by the capacity vector \(\cpc_t\). Elements of the vector represent available charging capacity in corresponding timeslots (time ranges in the green square). Possible charging session reservation requests arriving since the previous timestep is expressed by the vector \(\prd_t\), with ones representing the requested timeslots. Based on the three state variables \(\cpc_t, t, \prd_t\), the pricing policy provides an action \(a\), the price for charging, that the user either accepts (the first two states at the bottom) or rejects (the state on the right). 
    The state then transitions into the next timestep (details of the transition function are illustrated by \cref{fig:decision_tree}). The accepted charging request leads to reduced capacity values. The next charging session reservation is entered into the new state. Note that the timesteps must have much finer resolution than the charging timeslots. The gray color shows past information regarding the charging capacity and session vectors \(\cpc_t\) and \(\prd_t\), respectively.}
    \label{fig:cs_to_mdp_state}
\end{figure}

\begin{figure}
  \begin{center}
  \resizebox{0.99\textwidth}{!}{
    \tikzset{
  treenode/.style = {shape=rectangle, rounded corners,
                     draw, align=center,
                     top color=white, bottom color=blue!20},
  state_node/.style     = {treenode, font=\normalsize, bottom color=red!30},
  chance_node/.style      = {treenode, font=\normalsize},
  dummy/.style    = {circle,draw}
}

\begin{tikzpicture}
  [
    grow                    = down,
    level 1/.style           = {sibling distance = 7em}, 
    level 2/.style           = {sibling distance = 14em}, 
    level 3/.style           = {sibling distance = 7em}, 
    level distance          = 6em, 
    edge from parent/.style = {draw, -latex},
    every node/.style       = {font=\footnotesize},
    sloped
  ]
  \node (state_0) [state_node] {\(s_t=(\cpc, \prd)\)}
    child { node [chance_node] {\(p_{\text{acc}}(\prd,a) = 1-F_{\BDG_{\prd}}(a)\)}
        child { node [chance_node] {\(p_{\text{req}}(\prd') = \lambda_{\prd'}/h \)}
            child { node [state_node] {\(s_{t+1}=(\cpc', \prd')\)}
                edge from parent node [above] {\(p_{req}(\prd')\)}
                                 node [below] {\(\prd \leftarrow \prd'\)} }
            child { node {\(\cdots\)} }
            child { node [state_node] {\(s_{t+1}=(\cpc', \prd'')\)}
                edge from parent node [above] {\(p_{req}(\prd'')\)}
                                 node [below] {\(\prd \leftarrow \prd''\)} }
            edge from parent node [above] {\(p_{acc}\)}
                             node [below] {\(\cpc' \leftarrow \cpc-\prd\)} 
            }
        child { node {\(\vdots\)}
            edge from parent node [above, align=center] {\(1-p_{acc}\)}
                             node [below] {\(\cpc' \leftarrow \cpc\)}
            }
        edge from parent node [above] {action \(a\)}}
    child { node {\(\cdots\)} }
    child { node {\(\vdots\)}
        edge from parent node [above] {action \(a'\)}};
\end{tikzpicture}
  }
  \caption{
    The structure of the transition function \(\TR\). Given state \(s_t\), the probability of getting to the next state \(s_{t+1}\) is given by multiplying the probabilities along the edges. States are the decision nodes (in red), and chance states are in blue and contain the definition of the probability used along the edge. 
    }
  \label{fig:decision_tree}
  \end{center}
\end{figure}
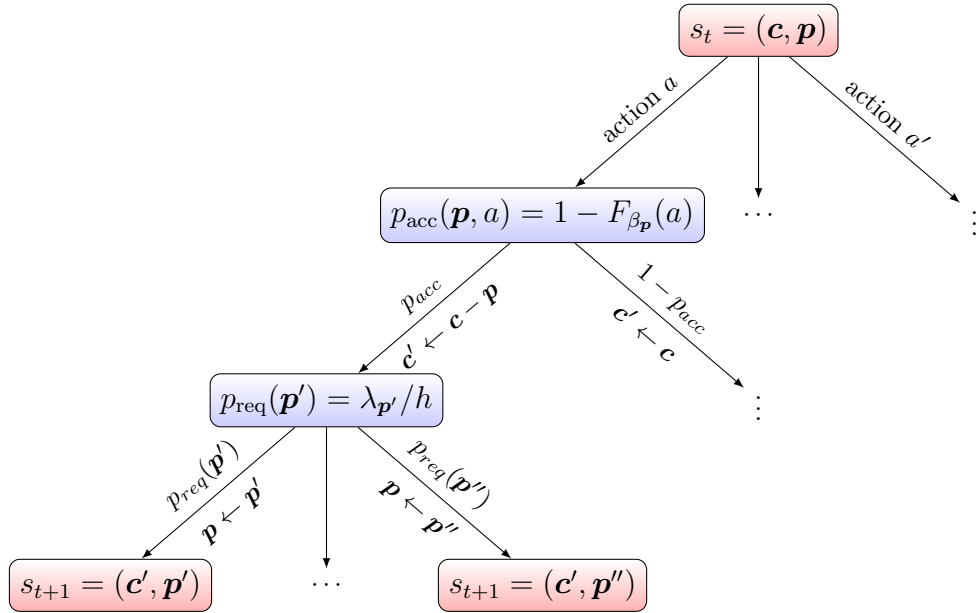

Having described the dynamic pricing problem, we will now develop the \MDP/ model to determine the pricing policy \(\pi\). 
The pricing policy is a mapping that assigns price \(a\) to the product-time pair \((\prd,\tau)\) combined with the currently available and planned future capacity \(c\)  of the seller's resources. 

In our \MDP/, defined as a 5-tuple \((\St, \TR, \Rw, \A, s_0)\),  the seller starts in some initial state \(s_0 \in \St\). Each state captures the current timestep, what product is being requested, and how many resources are currently available.  The seller offers a price \(a \in \A\) for the requested product, taking an action that results in a transition to a new state \(s' \in \St\). However, the transition is not deterministic because it is unknown whether the customer will accept or reject the price and what the next product request will be. By fitting the random demand process \(N(\tau)\) and distributions of the customer internal valuations \(\{\BDG_{\prd} | \prd \in \PRD \}\) to the historical data, we can estimate the transition probability \(\TR(s' | a, s)\), which determines the likelihood of reaching state \(s'\) when taking action \(a\) in the state \(s\). The transition between states also generates rewards for the seller, determined by the function \(\Rw(s, a, s')\). 

\subsubsection{State Space}
The \MDP/ \emph{state space} \(\St\) consists of states \(s = (\cpc, t, \prd)\) (we also use the notation \(s_t = (\cpc, \prd)\)). That is, the state is defined by the supply of all the resources \(\cpc\) at time step \(t\) and the product \(\prd\) being requested by some customer at time step \(t\). 
By discretizing the selling period \((0, \max(\vec{T}))\) into \(k\) time steps, we make sure the state space is finite. 

While continuous-time \MDP/ formulation~\cite{guoContinuousTimeMarkovDecision2009} is possible, it complicates the description of the problem, making it less intuitive and the solution more complex. 
The arrival time of customer product requests is continuous; the service provider can't influence these arrival times. However, they arrive as discrete events. 
Additionally, in pricing problems we are interested in, we assume that customer product requests arriving at similar times will mostly exhibit similar types of behavior. 
That is to say that we expect demand to depend on time in a piecewise continuous fashion with a finite number of discontinuities, where the discretization can be made to match these discontinuities. 
Additionally, as discussed in~\cref{sec:transition_poisson}, the discretization only needs to be fine enough to ensure that the probability of multiple requests arriving in a single timestep remains low.



The time step \(t\) increases by one with every transition, so \(s_t = (\cpc, \prd)\) is always followed by some \(s_{t+1} = (\cpc', \prd)\).

\subsubsection{Action Space}
The MDP \emph{action space} \(\A \subset \R\) is a set of possible prices that the seller can offer customers. This set can principally be continuous, but we assume the prices to form a finite set for our experiments. 


Similarly, as with continuous time, our model could accommodate for continuous action spaces~\cite{leeMonteCarloTreeSearch2020}. Doing so could improve the sellers' objectives, as the proposed prices could get closer to the internal customer's product valuation $b$. Most current service providers\footnote{Recent exception in the context of dynamic pricing in network revenue management is Lufthansa, which introduced continuous pricing in 2020~\cite{wenkertContinousPricing2020}. The reasons service providers use discrete prices in dynamic pricing are part technical (systems in airline revenue management have historically used discrete prices), part psychological/marketing to simplify customer choices.} selling directly to customers, however, offer discrete service prices. Therefore, it seems to be a natural simplification in our case as well. Furthermore, where possible, we compare our heuristic solution to optimal baselines only applicable to finite action spaces.

One special price is the infinite price $a=\infty$, which effectively forms the \emph{reject} action since no customer will have infinite funds. This action is used when the seller lacks the resources to provide the requested product.

\subsubsection{Reward Function}
The MDP \emph{reward function} \(\Rw(s_t, a, s_{t+1})\) determines the reward obtained by transitioning from \(s_t\) to \(s_{t+1}\) by taking action \(a\). If the seller's goal is revenue maximization, the reward is the price offered for the product. 
The reward has the value of the action \(a\) if the customer accepts the offered price \(a\) and \(0\) otherwise. Formally:
\[
  \Rw(s_t, a, s_{t+1})= 
    \begin{cases}
      a,  & \textnormal{if } \bdg > a\\ 
      0,  & \textnormal{otherwise}
    \end{cases}
\]
Here, \(a\) on the right-hand side of the equation is the value of the action \(a\). 
\(\bdg>a\) means that the customers budget \(\bdg\) is greater than the price \(a\), meaning customer accepts the price. 
Note that a successful sale implies capacity is reduced between \(s_t=(\cpc, \prd)\) and \(s_{t+1}\) from \(\cpc\) to \(\cpc-\prd\) in \(s_{t+1}\), which brings us to the definition of the transition function.



Depending on the CS operator's objectives, the reward function can have a different form, e.g., maximizing profits by subtracting the cost of providing the service from the revenue. 
In case of a public utility, the goal might be to maximize the utilization of the charging resources. 
However, utilization and revenue are not necessarily aligned, and optimizing both criteria simultaneously would require a multi-objective optimization approach.

\subsubsection{Transition Function}
The \emph{transition function} \(\TR(s_t, a, s_{t+1})\) is the most complex component of the \MDP/ model. It determines the state \(s_{t+1}\) the system develops into from state \(s_t\) when the service provider takes action \(a\). The transition function \(\TR\) is determined by two factors: the customer arrival processes \(\DEM(t)\) and the distributions of customer internal valuations \(\{\BDG_{\prd} | \prd \in \PRD \}\). The structure of the transition function and how it combines these two components is shown in \cref{fig:decision_tree}.

In some state \(s_t = (\cpc, \prd)\) (the root of the tree in \cref{fig:decision_tree}), the seller picks a price \(a\) for product \(\prd\) requested by a customer with a hidden internal valuation of the product. 
Since the customer accepts the offered price only if his internal valuation (modeled by a random variable \(X \sim \BDG_{\prd}\) of the product \(\prd\)) is greater than the offered price \(a\), the probability of a customer accepting the offered price is given by the complementary cumulative density function \(F_{\BDG_{\prd}}\) of \(\BDG_{\prd}\) as 
\begin{equation}
  \label{eq:p_acc}
  p_{\textnormal{acc}}(\prd, a) = P(X>a) = 1-F_{\BDG_{\prd}}(a)
\end{equation}
This is shown in the second level of the tree in \cref{fig:decision_tree}. The budget distribution \(\BDG_{\prd}\) could also be time-dependent, but we assume it is not for simplicity.

Independently of whether the product \(\prd\) is sold at time step \(t\) in \cref{fig:decision_tree}, some product \(\prd'\) could be requested at time step \(t+1\). 
The demand model \(D(t)\) that determines the probability of a product request at time step \(t\) is derived from the compound Poisson counting process \(N(\tau)\) with rate \(\lambda\).

The choice of the Poisson process to model arrival is a natural consequence of the so-called memoryless property that assumes that, at any point, the time until the next customer request does not depend on how much time has passed since the last customer request. Fortunately, this assumption holds true for many pricing problems since the assumed customer populations are large and customers act independently. For this reason, as well as its simplicity and useful properties, the Poisson process is a popular choice for modeling customer arrivals.

In our case, \(N(\tau)\) counts the arrival of a request for any product, and is created by merging independent, product-specific Poisson counting processes \(N_{\prd}(\tau)\) with rates \(\lambda_{\prd}\) into a single combined process. 

For the sake of explanation and without the loss of generality, we assume these processes are homogenous, i.e., the intensity \(\lambda\) does not depend on time.
Thus, the intensity of each product subprocess is \(\lambda_{\prd}\) and is constant, and from the properties of the Poisson process, we have \(\sum_{\prd\in\PRD}\lambda_{\prd} = \lambda\).

However, in our \MDP/ definition, we discretize the selling period \((0,T)\) into \(k\) intervals, the timesteps\footnote{\emph{Timesteps} are not to be confused with the \emph{timeslots} that form the resources used in the charging products, and that can have coarser discretization.}. Assuming for now that each timestep in the discretization has constant length \(\frac{T}{k}\), we approximate the Poisson process with a discrete demand process \(\DEM(t)\), \(t\in\{1,2,3,\ldots,k\}\). \(\DEM(t)\) gives the probability of product arrivals in each timestep. However, it allows for at most one product to be requested at any timestep. \(\DEM(t)\) is a multi-class extension of the Bernoulli process with \(|\PRD|+1\) possible outcomes, with \(+1\) for \emph{no} (empty) product request arriving in a timestep. The probabilities of the different products at timestep \(t\) in \(\DEM(t)\) are chosen in the following way: 
\begin{equation}
  \label{eq:p_req}
  p_{\textnormal{req}}(\prd, t) =   
  \begin{cases}
    \frac{\lambda_{\prd}}{k},                   & \prd \in \PRD\\
    1-\sum_{j \in \PRD}(\frac{\lambda_j}{k}),  & \prd=\emptyset
  \end{cases}
\end{equation}
See \cref{sec:transition_poisson} for details on how the discrete demand process with product request probabilities chosen this way behaves concerning the compound Poisson process \(N(\tau)\) and what is the quality of this approximation.

There is one issue when using the Poisson process with discretized time:  in discrete time, process \(\DEM(t)\) can generate at most \(1\) request in any interval of the discretization. In contrast, the continuous-time process \(N(\tau)\) can generate any positive number of requests in \emph{any} interval of non-zero length. 
Therefore, multiple charging requests can arrive at the same time step. 
If the requests arrive at the same time exactly, a randomly selected request will be processed first.
Otherwise, the requests are processed in the order of their arrival.
The first request is priced as usual, using the pricing policy for state \(s_t\). 
The second request is then simply priced with a policy based on the same timestep, in a state \(s'_{t}\), which corresponds to the same timestep, but possibly has reduced capacity.

Since the pricing policy is based on an \MDP/ formulation that does not consider multiple requests in the same timestep, even an optimal policy for this \MDP/ will result in suboptimal pricing in practice. 
However, by keeping the approximation error of the time discretization small, we can ensure that this situation is rare.
We consider the approximation of \(N(\tau)\) by \(\DEM(t)\) to be acceptable for a given number of timesteps \(k\) and the expected number of requests \(\lambda T\) under one condition: we want the continuous time process \(N(\tau)\) to rarely generate more than one request in any interval of the discretization.
Specifically, we want to pick \(k\) so that the error term \(\et{2}(k, \lambda)/\lambda\) is small. 
In the next section, we explain what we mean by \(\et{2}\) and justify our choice of approximation.


\section{Properties of the Demand Approximation in the MDP Model}
\label{sec:transition_poisson}


Here, we formalize the definition of the discrete demand process used in our \MDP/ model and quantify how well it approximates the assumed Poisson demand process. 

\subsection{Convergence of the Discrete Demand Process}
In this section, we will assume that the selling period \((0,T)\) is a unit interval \((0,1)\), which is without a loss of generality through simple rescaling of the timeline. Next, let us describe the Poisson demand process obtained by combining the Poisson sub-processes for each product. We assume there is a Poisson counting process \(N_{\prd}(\tau)\) for each product \(\prd\in\PRD\), defined by the rate \(\lambda_{\prd}\), generating arrival times of requests for that specific product. From the convenient properties of Poisson processes, the arrival times of all product requests can be considered as coming from a single compound Poisson process \(N(\tau)\) with intensity \(\lambda=\sum_{\prd\in\PRD}\lambda_{\prd}\). The compound Poisson process generates arrivals of requests for any product. 

However, since we discretize the time into \(k\) timesteps, we have to approximate the arrivals generated by the Poisson process by the discrete Bernoulli process. The Bernoulli process is a sequence of Bernoulli trials, defined by the number of trials \(k\) and the probability of arrival of any request in a single trial \(p\). The approximation is based on the fact that the Poisson process is a limit of a sequence of Bernoulli processes created by keeping the product \(kp=\lambda\) constant and taking \(k \rightarrow +\infty\). 

We can reconstruct the arrival process for individual products by assigning product indices to the arrivals according to the probabilities \(\frac{\lambda_1}{\lambda}, \frac{\lambda_2}{\lambda}, \ldots, \frac{\lambda_m}{\lambda}\). 
Arrival in the Bernoulli process can then be assigned a product type by sampling a discrete distribution with probabilities \(\frac{\lambda_1}{\lambda}, \frac{\lambda_2}{\lambda}, \ldots, \frac{\lambda_m}{\lambda}\), one for each product. We refer to the resulting object as \emph{discrete demand process} with \(m+1\) outcomes, where \(m\) outcomes correspond to the \(m\) possible products and \(m+1\) outcome corresponds to no request being made. We call this process the discrete demand process \(\DEM_{\vec{\lambda}}^k(t)\) in \(k\) time steps and demand intensities \(\vec{\lambda} = [\lambda_1, \ldots, \lambda_m]\). 

\begin{proposition}
\label{prop:convergence}
For given \(\vec{\lambda} = [\lambda_1, \ldots, \lambda_m]\), let \(\DEM^k_{\vec{\lambda}}(t)\) be a discrete demand process with \(k\) steps, \(k \ge \sum_{i=1}^m \lambda_i\), and \(m+1\) distinct possible values with outcomes \(i \in \{1, ...,m, \emptyset\} \) occurring with probability:
\[
p_i = 
\begin{cases}
  \frac{\lambda_i}{k}, & i \in \{1, ..., m\}\\
  1-\sum_{i=1}^m \frac{\lambda_i}{k}, & i = \emptyset
\end{cases}
\]  
Then, the sequence of the discrete demand processes \(\DEM^k_{\vec{\lambda}}(t)\) converges with \(k\rightarrow+\infty\) to the compound Poisson process with arrival intensity \(\lambda=\sum_{i=1}^m \lambda_i\) and discrete jump size distribution with probabilities \(\frac{\lambda_1}{\lambda}, \frac{\lambda_2}{\lambda}, \ldots, \frac{\lambda_m}{\lambda}\).
\end{proposition}

\begin{proof}
Consider the discrete demand process for a single product, which is a Bernoulli process \(B^k(t, p_i)\) with \(k\) steps and probability of success \(p_i=\lambda_i/k\). This is a sequence of \(k\) i.i.d. binary random variables with the probability of success \(p_i\). We know that with \(k\rightarrow +\infty\), the sequence \(B^k(t, \lambda_i/k)\) converges to a Poisson process with intensity \(\lambda_i\). Combining the \(m\) Poisson processes with \(\lambda_i, i\in \{1, \ldots, m\} \) gives a Poisson process with intensity \(\lambda=\sum_{i=1}^m\lambda_i\) with the required jump size distribution by the properties of combined Poisson processes. 

Then, the issue is whether these individual Bernoulli processes combine into the required discrete demand process \(\DEM^k_{\vec{\lambda}}(t)\) and whether this process converges to the same Poisson process as the combination of their limits. 

We show this using two Bernoulli processes consisting of i.i.d. random variables \(X_i\) and \(Y_i\), \(i\le k\). The success probability in the first process is \(p_1\) and \(p_2\) in the second process. These two processes can be combined into a new Bernoulli process with success probability \(p=p_1+p_2\),  and a random variable \(Z\) that determines the success class in the combined process with class probabilities \(\frac{p_1}{p}, \frac{p_2}{p}\).  However, unlike in the case of the Poisson process, the arrival classes are no longer independent since success in \(X_j\) implies failure in \(Y_j\). For separate processes, the event \(\{X_j=1, Y_j=1\}\) for some \(j\) has probability \(p_1p_2\). Nevertheless, since we have \(p_1 = \frac{\lambda_1}{k}, p_2 = \frac{\lambda_2}{k}\), the probability is \(P(X_j=1, Y_j=1)=\frac{\lambda_1\lambda_2}{k^2}\). Since \(P(X_j=1, Y_j=1)\) goes to \(0\) with \(\frac{1}{k^2}\) while \(p_1\) and \(p_2\) go to \(0\) with \(\frac{1}{k}\), the sequence of combined Bernoulli schemes approaches the desired Poisson process in the limit \(k\rightarrow+\infty\).
\end{proof}

To summarize, we use a discrete demand process in the \MDP/ model. The demand process consists of \(k\) die-rolls with \(m+1\) sided unbalanced die. Each die roll corresponds to a random arrival of a product request in some timestep. The \(m\) sides of the die correspond to the different product requests made in different timesteps, and the \(m+1\) side represents no request done by the customer in the given step. The proposition means that the demand process defined this way approximates the assumed naturally occurring Poisson arrival processes that generate arrivals of \(m\) different products on a real timeline.

\subsection{Approximation Quality}
\begin{figure*}
  \centering
  \includegraphics[width=0.99\textwidth]{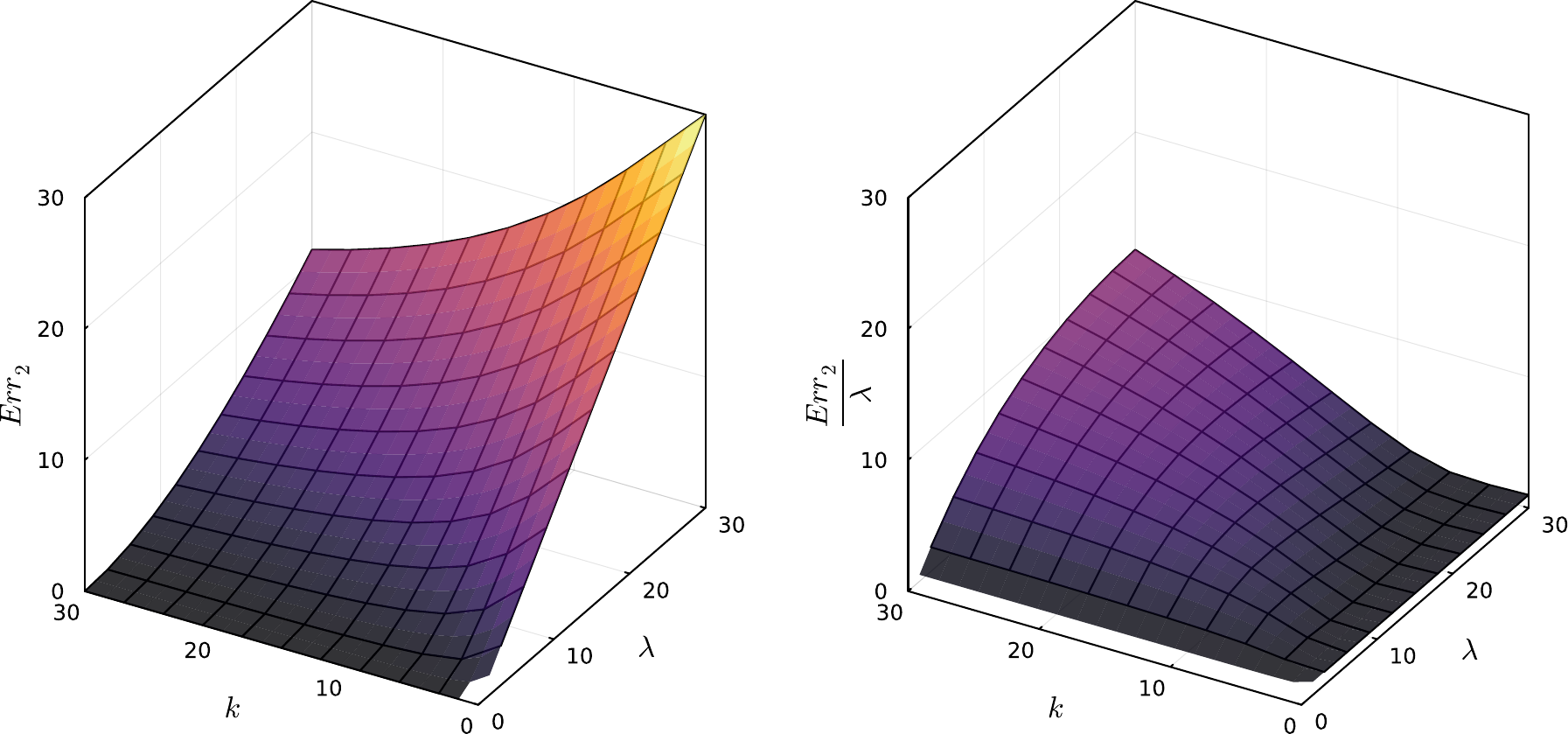}
  \caption{
    Visualization of the error term \(\et{2}(k,\lambda)\) (left) and the relative error \(\et{2}(k, \lambda)/\lambda\) (right) for different values of \(k\) and \(\lambda\). 
    \(\et{2}\) is defined in \cref{prop:err_formula} and represents the number of ignored events caused by the discretization of the Poisson process. With the number of timesteps \(k\) approaching \(0\), the number of ignored events approaches the expected number of events \(\lambda\). The relative error term \(\et{2}(k, \lambda)/\lambda\) (\cref{eq:relative_error}) represents the number of ignored events relative to the expected number of events.
  }
  \label{fig:demand_approximation_error_analytical}
\end{figure*}

The approximation quality depends on the value of \(k\) which represents the number of timesteps into which the continuous sale period \((0, 1)\) is divided. However, the discrete demand process with the probability of any arrival \(p = \sum_{i=1}^m \frac{\lambda_i}{k} \) is defined so that the expected number of requests in the \(k\) steps, which has a binomial distribution and expected value \(\EX(Bin(k, p)) = pk\), \emph{exactly} matches the expected number of arrivals from the Poisson process, \(\EX(Pois(\lambda)) = \lambda\), in the \((0,1)\) interval. Therefore, the Poisson and the approximating discrete processes do not differ in this metric.

Thus, the approximation error is more nuanced. By definition, the discrete demand process allows only for the arrival of the \(0\) or \(1\) request at each timestep. However, the Poisson process can generate more than one arrival in any real interval. This means that the discrete demand model systematically underestimates the probability of more than one arrival (effectively setting this probability to 0), while simultaneously overestimating the probability of exactly one arrival in every time step of the discretization. This overestimation and underestimation can be understood from the distribution of the number of arrivals in an interval of length \(1/k\) in the Poisson process. We will denote this random variable \(X\) and note that it is Poisson distributed\footnote{Unfortunately, the terminology here is historically misleading. In the discrete case, we have a sequence of Bernoulli-distributed random variables (sequence of coin tosses), which is called a Bernoulli process. The process's number of successes (heads) is distributed according to a binomial distribution. However, in the continuous case, we have a sequence of exponentially distributed random variables that form a Poisson process, and the number of arrivals in the process that is Poisson distributed.}:

\begin{align*}
  P(X=j) &= \frac{(\frac{\lambda}{k} )^j \e^{-\lambda/k}}{j!}
\end{align*}

We get the approximation error through the Taylor expansion of this term for different values of \(j\) corresponding to the probabilities of \(0, 1\) and \(2\) or more arrivals:

\begin{align*}
  P(X=0) &= \e^{-\lambda/k} = 1 - \frac{\lambda}{k} + o(\frac{\lambda}{k})\\
  P(X=1) &= (\frac{\lambda}{k}) \e^{-\lambda/k} = \frac{\lambda}{k} + o(\frac{\lambda}{k})\\
  P(X>1) &= o(\frac{\lambda}{k})
\end{align*}

Therefore, the approximation error in one timestep is in \(o(\frac{\lambda}{k})\) and decreases as \(k\) increases. However, we want to ensure not only that the error in one timestep is small but also that the accumulated error over all timesteps is acceptable. 

To do this, we define two error metrics. \(\et{1}\) is the expected number of intervals in the discrete demand process where the Poisson process would have more than one arrival. That is, the expected value of a binomial distribution \(Bin(p,k)\) with \(p\) being the probability of \(2\) or more arrivals from the Poisson process in an interval of length \(\frac{\lambda}{k}\).

However, \(\et{1}\) only gives an expected number of \emph{intervals} where we have a problem with more than one arrival, not the number of arrivals ignored by the discrete demand process, which could be higher. 
For this reason, we define \(\et{2}\), which gives the \emph{expected number of arrivals missed in all timesteps caused by ignoring additional (over one) arrivals in every timestep}. 
By symmetry argument from \(\EX(Pois(\lambda)) = \EX(Bin(k, \frac{\lambda}{k}))\), this is also the expected number of timesteps with one arrival added by the discrete demand process over the Poisson process.
\(\et{2}\) is formally given as:
\[
  \et{2} = \et{1}(k,\lambda) \EX_{Pois(\lambda/k)}[X-1 | X > 1]
\]
where \(X \sim Pois(\frac{\lambda}{k})\) is the Poisson distributed random variable that describes the number of arrivals in an interval of length~\(\frac{1}{k}\). \(\EX[X-1 | X > 1]\) is then the expected number of arrivals over one, conditional on more than one arrival.

The following proposition gives formulas for these two kinds of errors.

\begin{proposition}
\label{prop:err_formula}
Let \(\DEM^k_{\vec{\lambda}}(t)\) be a Bernoulli demand process that converges to a Poisson process with intensity \(\lambda\) on a unit interval when the discretization of time into \(k\) timesteps is refined in the limit \(k\rightarrow +\infty\). 

Then, the expected number of timesteps in which the Poisson process will register more than one arrival is
\begin{align}
  \et{1}(k,\lambda) &= k - (k+\lambda) \e^{-\lambda/k} \label{eq:err1_formula}
\end{align}

The expected number of arrivals over one, summed across all timesteps, is
\begin{align}
  \et{2}(k,\lambda) &= \lambda \e^{-\lambda/k} + (\lambda - k) (1-\e^{-\lambda/k}) \label{eq:err2_formula} 
\end{align}

\end{proposition}
\begin{proof}
\Cref{eq:err1_formula} is easily seen from the definition. The expected value of \(Bin(p,k)\) is \(kp\), substituting the probability of two or more events in Poisson as \(p = 1 - P(X=0) - P(X=1) = 1-\e^{-\lambda/k}-\frac{\lambda}{k} \e^{-\lambda/k}\) and simplifying the expressions immediately provides the result. 

To show \cref{eq:err2_formula}, we need to express the conditional expectation
\[\EX_{Pois(\lambda/k)} \left[ X-1 | X > 1\right] \] 
of a Poisson distributed random variable \(X\sim Pois(\frac{\lambda}{k})\):

\begin{align}
  &\EX \left[ X-1 | X > 1\right] = \sum_{j=0}^{+\infty} j P(X-1=j | X>1) \\
  &= \sum_{j=0}^{+\infty} j \frac{P(X-1=j \wedge X>1)}{P(X>1)} \\
  &= \frac{1}{1-P(X=0)-P(X=1)} \sum_{j=2}^{+\infty}(j-1)P(X=j) \\
  &= \frac{(\frac{\lambda}{k} - \lambda \e^{-\lambda/k}) - (1- \e^{-\lambda/k} - \lambda \e^{-\lambda/k})}{1- \e^{-\lambda/k} - \lambda \e^{-\lambda/k}} 
\end{align}
We use the fact that \(\EX\left[Pois(\frac{\lambda}{k})\right] = \frac{\lambda}{k}\). 
We get the desired term by multiplying the last line by \(\et{1}\) and simplifying. 
\end{proof}

In \cref{prop:convergence}, we show that the discrete demand process converges to the Poisson process when refining the discretization of time. However, it does not tell us about the quality of the approximation for a given \(k\). \Cref{prop:err_formula} explicitly quantifies the two kinds of errors in the approximation in terms of the number of timesteps \(k\) and the expected number of requests \(\lambda\). 
Since the error depends on the expected number of requests \(\lambda\), in experiments, we will use the relative error 
\begin{align}
  \label{eq:relative_error}
  \frac{\et{2}(k, \lambda)}{\lambda}
\end{align}
to quantify the quality of the approximation. The interpretation of this relative error is \emph{number of missed additional requests per expected request}. As mentioned above, in the discrete model, these requests are not actually missed; they show up as ``hallucinated'' additional requests at different timesteps.

\section{Dynamic Pricing Algorithm Using MCTS}
\label{sec:msts_pricing_algorithm}

\begin{algorithm}
  \caption{Dynamic pricing MCTS algorithm for MDPs. Based on \cite{mausamPlanningMarkovDecision2016,egorov2017pomdps}.}
  \label{alg:mcts}

\begin{algorithmic}[1]
    \Procedure{MCTS}{\(\currentState; c, \depthLimit, \iterLimit\)}
    \State $n \gets 0$
    \While{ \(n < \iterLimit\)}
        \State $n \gets n+1$
        \State $s_0\gets\currentState$ 
        \State $r_{0}\gets0$
        \State $\maxReachedDepth\gets0$
        \For{$i=1$ to \(\depthLimit\) } \Comment{Selection-Expansion loop to max. tree depth \(\depthLimit\)} \label{alg:mcts_selection_loop_start}
            \If{$s$ not encountered yet} \Comment{Expansion} \label{alg:mcts_expansion}
                \State $n_{s}\gets0$
                \For{$a\in\A$}
                    \State $n_{s,a}\gets 0$
                    \State $q_{s,a}\gets 0$
                \EndFor
            \EndIf
            \State $\maxReachedDepth \gets i$
            \If{\(\exists a \in \A \) s.t. \(n_{s,a} = 0\)} \Comment{Avoids undefined operation on line \ref{alg:mcts_ucb}}
                \State \(a_i \gets a\) \label{alg:mcts_first_action}
            \Else
                \State $a_i\gets\argmax_{a\in\A} q_{s,a} - c\sqrt{\frac{\ln(n_{s})}{n_{s,a}}}$ \Comment{Selection (using UCB1)} \label{alg:mcts_ucb}
            \EndIf
            \State $s'\gets  \TR(s'|a_i,s) $  \Comment{Sample execution of action $a$ in $s$} \label{alg:mcts_simulate}
            \State $r_{i}\gets r_{i-1}+R(s,a_i,s')$ \Comment{Cumulative reward received up to iteration \emph{i}} \label{alg:mcts_reward}
            \State $s_{i}\gets s$
            \State $s\gets s'$
            \If{$s$ terminal \textbf{or} \(n_{s,a_i}=0\)} \label{alg:mcts_terminal_condition}
                \State \textbf{break}
            \EndIf
        \EndFor\label{alg:mcts_selection_loop_end}
        \State $r_{\maxReachedDepth}\gets r_{\maxReachedDepth-1} + \Call{rollout}{s}$ \Comment{Value estimation}\label{alg:mcts_rollout}
        \For{$i=1$ to $\maxReachedDepth$} \Comment{Backpropagation}\label{alg:mcts_backpropagation}
            \State $q_{s_i,a_i} \gets \frac{n_{s_i, a_i} q_{s_i,a_i} + (r_{\maxReachedDepth} - r_{i-1})}{n_{s_i, a_i}+1}$ \Comment{Average \(q_{s,a}\) with reward from levels below} \label{alg:mcts_bacprop_q_val_update}
            \State $n_{s_i}\gets n_{s_i}+1$
            \State $n_{s_i,a_i}\gets n_{s_i,a_i}+1$
        \EndFor
    \EndWhile
    \State \textbf{return} $\argmax_{a\in\A} q_{s_0, a}$ \label{alg:mcts_action}
    \EndProcedure
\end{algorithmic}

\end{algorithm}

\begin{algorithm}
  \caption{Rollout algorithm used in \cref{alg:mcts}}\label{alg:rollout}
  \begin{algorithmic}[1]
    \Procedure{Rollout}{$\text{state}=(\cpc, t, \prd)$}
    \State $s \gets \text{state}$
    \State $r \gets 0$
    \While{$s$ is not terminal}
        \State $a'\gets $ select random action from $\A$
        \State $s'=(\cpc', t', \prd')\gets \TR(s'|a_i,s) $ \Comment{sample result of action $a'$ in $s$}
        \State $r\gets r+R(s, a', s')$ 
        \State $\Delta t \gets $ sample inter-arrival time distribution of demand at timestep $t$ \label{alg:rollout_interarrival_time}
        \State $\prd' \gets$ sample product request at $t+\Delta t$
        \State $s\gets(\cpc', t+\Delta t, \prd')$
    \EndWhile
    \State \textbf{return} $r$ 
    \EndProcedure
\end{algorithmic}

\end{algorithm}

This section describes the method we use to derive the dynamic pricing policies. Our solution method of choice for large-scale problems is \MCTS/. 
Unlike tabular methods, such as \gls{vi} or policy iteration, \MCTS/ does not need to enumerate the whole state space. 
Instead, it looks for the best action from the current state and expands only to states that the system is likely to develop into. 
However, unlike VI, for every state, \MCTS/ only approximates best actions. 
\MCTS/ improves its approximations of best action with the number of iterations. 

Nonetheless, it can be stopped at anytime to provide currently the best approximation of optimal action. These properties make it a helpful methodology in dynamic pricing. With \MCTS/, we can quickly apply changes in the environment to the solver. Even in large systems, the price offer can be generated quickly enough for a reasonable customer response time. To the best of our knowledge, this is the first attempt to solve the \EV/ charging dynamic pricing problem using \MDP/ and~\MCTS/. 


In our \MCTS/ implementation, we use the popular Upper Confidence-bound for Trees (UCT) variant of \MCTS/~\cite{auerFinitetimeAnalysisMultiarmed2002,coulomEfficientSelectivityBackup2006,mausamPlanningMarkovDecision2016} that treats each node as a bandit problem and uses the upper confidence bound formula to make the exploration-exploitation trade-off. 
We split the algorithm presentation into two parts. 
The first part covers the selection and expansion of tree nodes that form the tree policy and backpropagation. 
This part of the algorithm is shown in \cref{alg:mcts}. The second part of the algorithm is the rollout policy, which is shown in \cref{alg:rollout}.

\subsection{Tree Policy and Backpropagation}

The input of \cref{alg:mcts} is the current state for which we seek to estimate the best action \(a\). The algorithm has three parameters, the exploration constant \(c\), the number of iterations \(\iterLimit\), and the tree depth limit \(\depthLimit\). 

The input state corresponds to a tree's root from which the \MCTS/ algorithm builds the search tree. Each tree node corresponds to a state \(s\) of the \MDP/, and for each node, we keep track of how many times was the given node visited, \(n_s\), how many times was action \(a\) used in given state, \(n_{s, a}\) and a running average of the q-value of each state-action pair, \(q_{s, a}\). These values are iteratively updated during the run of the algorithm.

Each iteration of the \MCTS/ algorithm works in 4 stages: selection, expansion, value estimation, and backpropagation. The existing tree nodes are first traversed in the \emph{Selection} step (Lines \ref{alg:mcts_selection_loop_start}-\ref{alg:mcts_selection_loop_end}).
Each tree node selects an action using the UCB formula (Line \ref{alg:mcts_ucb}). The exploration constant \(c\) sets the appropriate exploration-exploitation trade-off between selecting actions with high q-value averages and actions that were not yet tried often. 
The node to transition to is determined by sampling the \MDP/ transition function \(\TR\) for the outcome of action \(a\) in the state \(s\) (Line \ref{alg:mcts_simulate})\footnote{The description of the algorithm may lead one to expect that during the selection phase of the algorithm, each selection leads to a new state and one level lower in the tree. This is not necessarily true for all \MDP/s since the transition from state \(s\) in \(i\)th step in the selection loop can result in a state encountered in some previous step. However, this cannot happen in our finite-horizon \MDP/ as it includes timestep as a state variable.}. 
If the resulting state has already initialized tree node variable \(n_s\), the same process is repeated from this node. This process is repeated for \(\depthLimit\) iterations or until a node corresponding to a terminal state is reached or until a state that does not yet have a corresponding tree node is encountered. 

When a previously unseen state is encountered, the tree is \emph{Expanded} with a new node corresponding to this state (Line \ref{alg:mcts_expansion}). Additionally, the stat-action counters and q-values are initialized for this node for all possible \(a\in\A\). However, for the freshly expanded node, the exploration term in the UCB formula on Line \ref{alg:mcts_ucb} has an undefined value. Therefore, if there is an untried action in some state \(s\), we first try any unused action \(a\) with \(n_{s,a}=0\) (Line \ref{alg:mcts_first_action}). Additionally, when we encounter an untried action, we terminate the selection-expansion stage, even if the tree depth limit has not been reached yet (Line \ref{alg:mcts_terminal_condition}). As a result, our MCTS implementation proceeds in a breadth-first manner.

During the selection phase of the algorithm, the cumulative reward collected up to \(i\)th iteration is recorded (Line \ref{alg:mcts_reward}). The estimated value of the newly explored state-action pair \(q_{s,a}\) is then calculated in the \emph{Value estimation} stage of the iteration using the rollout (Line \ref{alg:mcts_rollout}). 

The fourth stage of the \MCTS/ iteration is the \emph{backpropagation} (Lines \ref{alg:mcts_backpropagation}). Here, the cumulative rewards are used to update the average q-value of the node on the \(i\)th level of the tree, \(q_{s_i,a_i}\). The value is updated with \(r_d - r{i-1}\), \emph{all} reward collected \emph{below} the \(i\)th level in the tree, including the reward from the rollout. The update is done by averaging the previous update values with the new value (Line \ref{alg:mcts_bacprop_q_val_update}).

The four stages (selection---expansion---value estimation---backpropagation)  are repeated as many times as allowed. Finally, when the main loop terminates, the algorithm estimates the best action based on the q-value of actions in the tree's root (Line \ref{alg:mcts_action}). 

\subsection{Rollout Policy}

\Cref{alg:rollout} presents the second part of the \MCTS/ algorithm, the rollout. It is applied from state-action pairs that were used for the first time in the tree policy. The rollout approximates the reward of the selected action by quickly reaching the terminal state. In our experiments, we use the uniformly random rollout policy that applies random actions until a terminal state in the \MDP/ is reached. 

Because we use the Poisson process (Bernoulli processes after discretization) as the customer arrival process, we can speed up the rollout by sampling the time to the next arrival from the inter-arrival distribution. It has a geometric distribution in the discretization (\cref{alg:rollout_interarrival_time} in \cref{alg:rollout}). 
Therefore, we can arrive at the terminal state in fewer steps. 
In the rollout, we simulate actions without storing their resulting states until reaching a terminal state. 
When the terminal state is reached, the rollout terminates immediately, returning the accumulated reward.

\subsubsection{Implementation}
\label{sec:mcts_implementation}
Our implementation is based on the MCTS implementation in the POMDPs.jl\cite{egorov2017pomdps} library that uses recursion when traversing the tree, unlike the description in \cref{alg:mcts}. We provide the unrolled iterative description for clarity.

The implementation we use reuses the constructed decision tree between the steps of the ``real'' \MDP/ that happen outside of the \MCTS/ algorithm, improving the convergence speed. 
In our experiments, we build the tree to the maximum depth \(\depthLimit=3\) with the exploration constant set to \(c=1\). The number of iterations is capped at \(\iterLimit=800\). We find that these low numbers are sufficient for good performance in our experiments and result in a reasonable computation time.

\section{Experiments and Results}
\label{sec:experimental_results}
We compare our \MCTS/ pricing solution against multiple baselines on artificially generated problem instances modeled on a real-life CS dataset provided by a local German CS operator.
The dataset contains \num{1513} charging sessions from a single charging location spanning \num{2} years of operations, with an average charging duration of \SI{44}{minutes}. The dataset contains the start time of the charging session and the duration of charging.
The dataset is primarily used to inform the choice of distributions used in the generated problem instances.
The histograms of charging session start times and durations are shown in \cref{fig:hist_dataset}. 

As discussed in \cref{sec:introduction}, the proposed dynamic pricing method is developed with high-demand CSs in mind.
However, the available dataset is from a low-demand CS. 
Therefore, when creating problem instances for the experiments, the demand parameters are set to be much higher than the ones in the dataset, as described in \cref{tab:instance_parameters}.
Our experiments show that the kind of dynamic pricing we propose is not as useful in low-demand situations and is likely not worth the associated overhead. 
Additionally, we assume a highway charging location with multiple fast chargers, where EV drivers depart the CS voluntarily after their reserved time ends. 




The main goal of the experiments is to demonstrate the viability of the \MCTS/ dynamic pricing algorithm for \EV/ charging. To show this, we run a number of simulations with problem instances created using different problem parameters and compare the average performance of the \MCTS/ algorithm with the baseline methods. 

\subsection{Problem Instances}
\begin{table}[h]
  \centering
  \caption{Overview of the problem instance parameters.}
  \label{tab:instance_parameters}
  \centering
  \begin{tblr}{
    colspec={lXc},
  }
    \hline
    Instance property & Structure & Parameters\\
    \hline
    \textbf{Fixed parameters} & &\\
    Request length [h] & Exponential distribution & $\theta=3$\\
    Request start time [h] & Normal distribution & $\mu=12$,  $\sigma=3$ \\
    Customer budget [$1/$h] & Normal distribution & $\mu=1$,  $\sigma=0.5$\\
    Instance duration [h] & & 00:00 to 23:59\\
    \hline
    \textbf{Varied parameters} & & \\
    Request arrivals [\#/\SI{24}{\hour}] & Discretized Poisson proc. & 6 to 288\\
    Timeslot length [h] & discrete value range & \SI{15}{\minute} to \SI{12}{\hour}\\
    Timestep length [h] & discrete value range & \SI{2}{\minute} to \SI{1.5}{\hour}\\
    \hline 
  \end{tblr}
\end{table}

\begin{figure*}
  \centering
  \begin{subfigure}[t]{.4\textwidth}
    \centering
    \includegraphics[width=\textwidth]{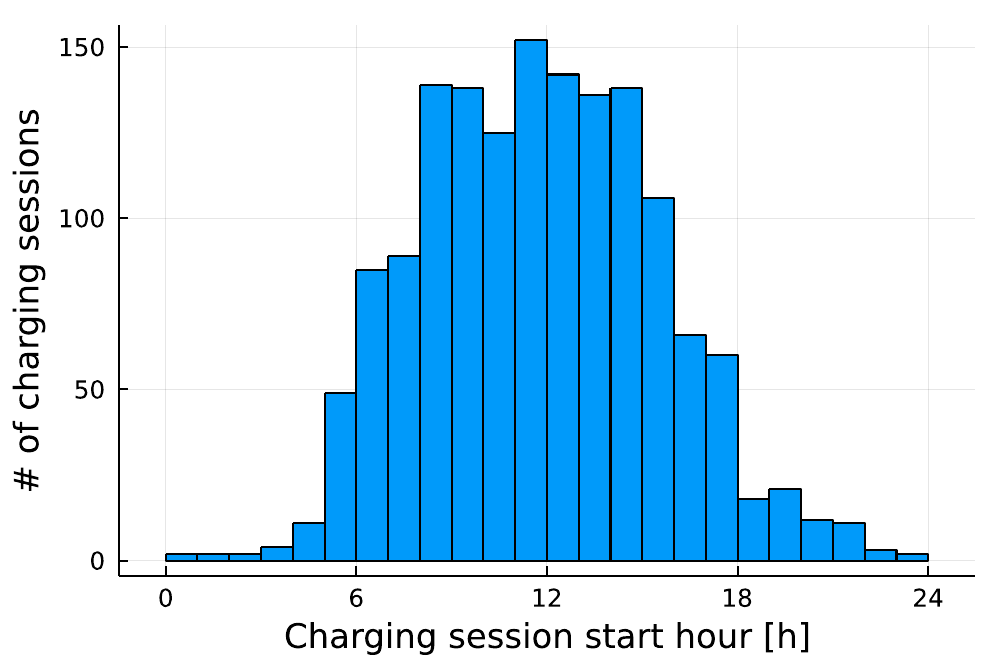}
    \caption{}
    \label{fig:hist_dataset:start_hour}
  \end{subfigure}
  \begin{subfigure}[t]{.4\textwidth}
    \centering
    \includegraphics[width=\textwidth]{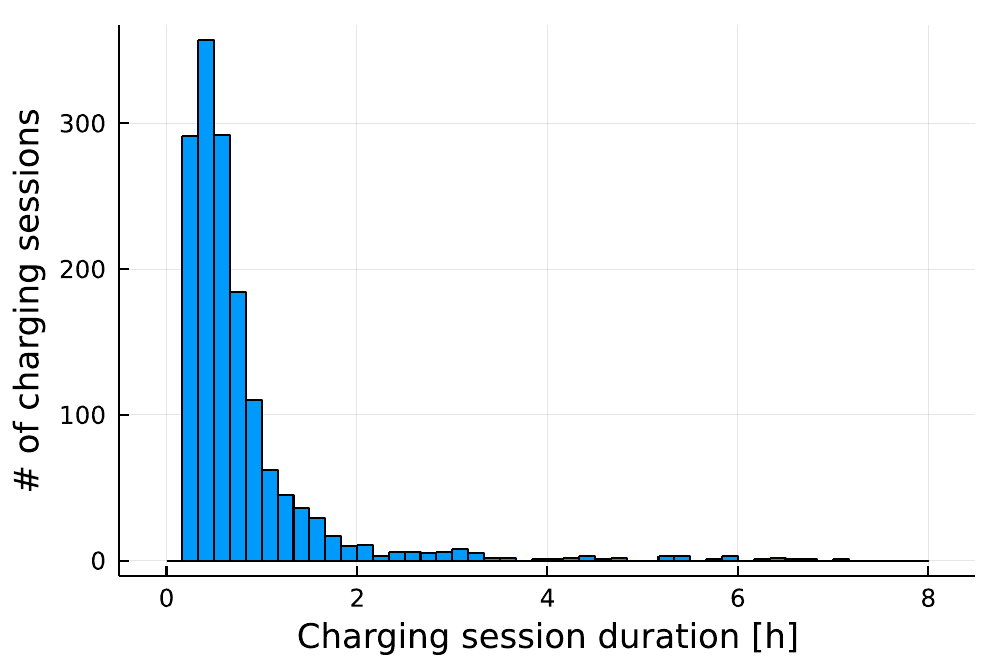}
    \caption{}
    \label{fig:hist_dataset:duration} 
  \end{subfigure}
  \caption{Histograms show the start (a) and duration (b) of charging sessions contained in the dataset that was used to select distributions for experiments.}
  \label{fig:hist_dataset}
\end{figure*}

While the generated instances are simple, the approach is flexible, and it can accommodate much more complicated inputs. We use basic distributions with parameters fitted to our  \EV/ charging dataset. The user budgets are sampled from the user budget distribution for the given charging session length. As we did not have any data on user budgets, the choice of parameters is arbitrary.  
We set the charging session start times to follow the normal distribution and the charging session length to follow the exponential distribution. We set both distributions to be independent. While this is not true in practice, it is sufficient to showcase our method and make the definition of the instances simpler. 

The request arrival times are obtained from a discretized homogeneous Poisson process, as described in \cref{sec:transition_poisson}. 
Time-variable Poisson demand process intensity that appears in the real world would be handled by a simple transformation of the timeline, resulting in discrete variable-length timesteps, which lead to a constant probability of request arrival in the discretized arrival process. 
Therefore, using a homogenous Poisson demand process in experiments is without the loss of generality of the method.  
Where available, we set their parameters based on the data analysis of our dataset; these values are collected in \cref{tab:instance_parameters}. 

Each problem instance has a form of a charging request sequence, as shown by \cref{eq:request_sequence}. 
Each pricing method, including the baseline methods described below, then prices the requests in the order they come in. 
The accumulated reward is averaged across \num{100} simulated runs to measure the performance of each method for comparison with other methods. 

\subsection{Baseline Methods}
Because of the difficulties of evaluating the dynamic pricing policies, we evaluate our proposed \MCTS/ solution against three baseline methods: \emph{flat rate}, \emph{\MDP/-optimal VI}, and \emph{Oracle} pricing methods. The flat rate represents the lower bound on the revenue we might expect from a dynamic pricing solution. The~VI baseline returns an optimal pricing policy and represents the best possible pricing method for the \MDP/ model. Finally, the Oracle policy represents the unachievable upper bound on dynamic pricing performance. Oracle provides the best possible allocation by assuming the CS operator has a perfect knowledge of future requests and \EV/ users' budgets, which is unrealistic in real-world use cases.

\subsubsection{Flatrate}
This baseline provides a lower bound for our \MCTS/ pricing method and a reference for showing how much improvement dynamic pricing could bring. 
The flatrate baseline is a single flat price per minute of charging that is used for all charging requests.

The flat price is determined from a set of ``training'' pricing problem instances that take the form of sequences shown in \cref{eq:request_sequence}. 
In each problem instance, we evaluate every possible flatrate price that corresponds to an action in the \MDP/ and measure the resulting revenue. 
The price that maximizes the average revenue across all training sequences is then used as the flatrate price in the testing simulation runs. 
This pricing method still uses reservations and the reservations are allocated in in sequence, resulting in first-requested, first-reserved allocation.

As discussed in \cref{sec:other_related_work}, the flat rate per hour or \unit{\kWh} are the most common pricing strategies. Since we make the simplifying assumption that all the charging sessions use the same constant charging power, the flatrate baseline represents both of the flat rate pricing strategies.

\subsubsection{Value Iteration (VI)} 
Our second baseline pricing method is the optimal \MDP/ policy generated by a VI algorithm~\cite{russellArtificialIntelligenceModern1995}. 
VI is a simple yet accurate method for solving \MDP/s that converges to an optimal policy for any initialization. 
The advantage of VI is that it quickly converges to a complete near-optimal pricing policy at the cost of enumerating the whole state space in memory. 

Based on the structure of our \MDP/ state, the state-space size of our \MDP/ model is \(k c_0^n 2^n\), where \(k\) is the number of timesteps, \(n\) is the number of charging timeslots and \(c_0\) is the initial charging capacity. 
If we limit the reservations to contain only the contiguous timeslots, as we do, the state-space size reduces to \(k c_0^n n(n+1)/2\). 
This gives VI an exponential space complexity in the number of timeslots. Thus, it does not scale well to larger problem instances. 
Therefore, we use VI only to obtain optimal policies on smaller problem instances to validate the heuristic approach of \MCTS/. 

Note that there are other exact solution methods for \MDP/ problems than VI, such as policy iteration or linear programming. All these methods can provide the same optimal pricing policy as VI. However, just like VI, all these methods require enumeration of the whole state space. Our choice of VI is, therefore, arbitrary in this sense. 

Since VI gives optimal pricing policy, we use it to benchmark the performance of our \MCTS/ approach, which, since it is heuristic, is expected to provide worse results. How much worse is the question we want to answer by comparing the performance of the two methods on small-enough instances that VI can solve.

\subsubsection{Oracle} Finally, we compare our \MCTS/-based pricing method against the \emph{Oracle} baseline strategy. It's important to note that the Oracle strategy, while used for benchmarking, isn't practically applicable due to its nature. Unlike other pricing strategies, Oracle relies on having prior knowledge of the entire request sequence and the budgets of \EV/ users to determine prices. 

Using this knowledge, Oracle maximizes the optimization metric to provide a theoretical upper bound on the revenue and resource usage achievable by any pricing-based allocation strategy. It works for large and small problem instances; therefore, we can use it to track the performance of \MCTS/ across a wide range of problem sizes.

The Oracle pricing solution is obtained from a linear program. For \(k\)th sequence of charging requests \(\CST_k\) with requests indexed by \(i\), the~optimum revenue is the result of a simple binary integer program:
\begin{align}
  \textnormal{maximize}   \displaystyle\sum\limits_{i \in \{1\ldots |\CST_k|\}} & x_i  \floor*{b_i}_A, \,\, \textnormal{subject to:} &   \label{eq:LP_obj} \\
  \displaystyle\sum\limits_{i \in \{1\ldots |\CST_k|\}} & x_i  \prd_i^j  \le \cpc_0^j  & j=1 ,...,  |\RES| \label{eq:LP_cpc} \\
                          & x_i \in \{0,1\}  & i=1 ,..., |\CST_k|. \nonumber         
\end{align}
where, \(x_i\) are the binary decision variables that determine which requests from \(\CST_k\) are accepted by the CS operator. In~the objective function~(\cref{eq:LP_obj}),  the~term \(\floor*{b_i}_A = \max_{a \in A, a~\le b_i} a\) denotes the fact that the budget values in the sequence \(\CST_k\) are mapped to the closest lower values in the action space \(A\).  Conditions~(\cref{eq:LP_cpc}) mean that the accepted charging sessions have to use fewer resources than the initial supply \(\cpc_0\). 

\subsection{Results}

\definecolor{area_background}{rgb}{0.9,0.9,0.9}
\begin{sidewaystable}[htpb]
  \centering
  \caption{Selected results of the experiments. Values are averaged from 100 instances.}
  \label{tab:results}
  \centering
  \begin{tblr}{
    colspec={lrrrrrrrr},
    cell{2-7}{1-Z} = {area_background},
    cell{14-19}{1-Z} = {area_background}
  }
  \hline
  Method & {Timeslot \\ len. {[min]}} & {Timestep \\ len. {[min]}} & Revenue {[1]} & \SetCell{r} Utilization {[h]} & \SetCell{r} {Average num.\\ of requests} & \SetCell{r} {Average num. of \\ accepted requests} & \SetCell{r} {Average Instance \\Runtime [s]} \\
  \hline
  \SetCell[r=6]{l} Oracle &  360 & 15.00 & 67.45 & 53.58 & 18.98 & 8.61 & 0.002 \\
  Oracle & 180 & 7.50 & 64.31 & 51.78 & 36.67 & 14.82 & 0.002 \\
  Oracle & 120 & 5.00 & 62.66 & 51.72 & 49.95 & 20.14 & 0.002 \\
  Oracle & 60 & 2.50 & 57.91 & 50.59 & 76.28 & 29.81 & 0.003 \\
  Oracle & 40 & 1.67 & 55.58 & 50.49 & 92.86 & 35.73 & 0.005 \\
  Oracle & 20 & 0.83 & 51.67 & 49.52 & 123.96 & 45.31 & 0.007 \\
  \SetCell[r=6]{l} Flatrate & 360 & 15.00 & 39.33 & 48.72 & 18.98 & 7.74 & 0.001 \\
  Flatrate & 180 & 7.50 & 38.36 & 47.52 & 36.67 & 13.16 & 0.001 \\
  Flatrate & 120 & 5.00 & 40.14 & 45.96 & 49.95 & 17.35 & 0.001 \\
  Flatrate & 60 & 2.50 & 38.49 & 44.07 & 76.28 & 25.75 & 0.002 \\
  Flatrate & 40 & 1.67 & 38.27 & 43.82 & 92.86 & 30.91 & 0.003 \\
  Flatrate & 20 & 0.83 & 36.64 & 41.96 & 123.96 & 43.84 & 0.008 \\
  \SetCell[r=6]{l}  MCTS & 360 & 15.00 & 44.01 & 45.48 & 18.98 & 7.15 & 1.467 \\
  MCTS & 180 & 7.50 & 43.71 & 44.85 & 36.67 & 12.01 & 5.012 \\
  MCTS & 120 & 5.00 & 43.15 & 45.22 & 49.95 & 16.12 & 10.905 \\
  MCTS & 60 & 2.50 & 40.68 & 44.20 & 76.28 & 23.59 & 39.413 \\
  MCTS & 40 & 1.67 & 39.99 & 44.13 & 92.86 & 27.86 & 93.557 \\
  MCTS & 20 & 0.83 & 38.51 & 42.75 & 123.96 & 38.75 & 450.417 \\
  \SetCell[r=3]{l} VI & 720 & 30.00 & 48.18 & 49.44 & 9.00 & 4.11 & 0.002 \\
  VI & 480 & 20.00 & 43.88 & 42.48 & 13.94 & 5.16 & 0.002 \\
  VI & 360 & 15.00 & 46.23 & 44.76 & 18.98 & 7.06 & 0.002 \\
  \hline
  \end{tblr}
\end{sidewaystable} 

We evaluate the proposed \MCTS/ pricing method in the experiments against the baselines described in \cref{sec:experimental_results}. 
First, we analyze the performance of the \MCTS/ algorithm in a set of small instances to select the best hyperparameters. 
Then, we compare the performance of the \MCTS/ algorithm with the flatrate, VI, and Oracle baselines in a set of different instances, ranging from instances with a small state space where VI can still generate results to instances with a larger state space.

\subsubsection{MCTS Hyperparameter Grid Search}

\begin{figure}
  \centering
  \includegraphics[width=0.99\textwidth]{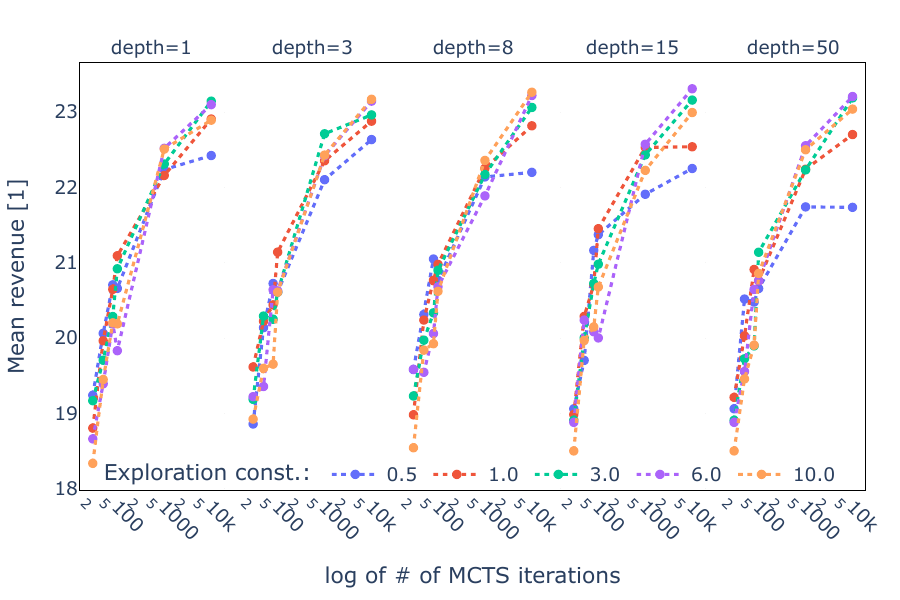}
  \caption{
    Figure showing the \emph{performance} of the \MCTS/ algorithm on the same problem instance with varied hyperparameters. The plots show means from \num{100} different request sequences. 
  } 
\label{fig:mcts_gridsearch}
\end{figure}

\begin{figure}
  \centering
  \includegraphics[width=0.99\textwidth]{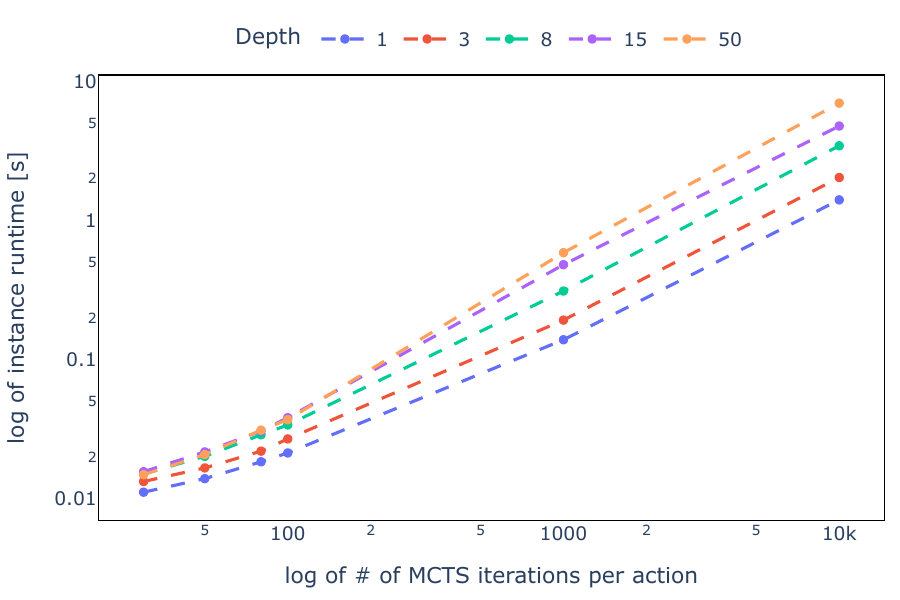}
  \caption{
    Figure showing the \emph{runtime} of the \MCTS/ algorithm on the same problem instance with varied hyperparameters and exploration constant fixed at \num{0.3}. Changing the exploration constant has a negligible effect on the resulting runtime. The plots show means from \num{100} different request sequences.
  } 
  \label{fig:mcts_runtime}
\end{figure}

The \MCTS/ algorithm (\cref{alg:mcts}) has three parameters, the number of iterations \(\iterLimit\), the exploration constant \(c\), and the tree depth limit \(\depthLimit\). 
To determine its hyperparameters, we use a modestly sized instance with \num{12} timeslots and \num{96} timesteps, three CSs, and expected \num{24} charging requests. 

We perform a full grid search with discretized sequences of these parameters. The results of the grid search are shown in \cref{fig:mcts_gridsearch}. 
The results are averaged from \num{100} sampled request sequences. 
The main takeaway from the grid search is that the number of iterations has the most significant impact on the algorithm's performance. 
The exploration constant and the tree depth limit have a much smaller impact. 
However, the number of iterations and the tree depth have the highest impact on the computation time, as seen in \cref{fig:mcts_runtime}.
Therefore, in the following experiments, we set the number of iterations to \(\iterLimit=10000\), the tree depth to \(\depthLimit=10\), and the exploration constant to \(c=3\).

\subsubsection{Varying Timeslot Length}

\begin{figure}
  \centering
  \includegraphics[width=0.99\textwidth]{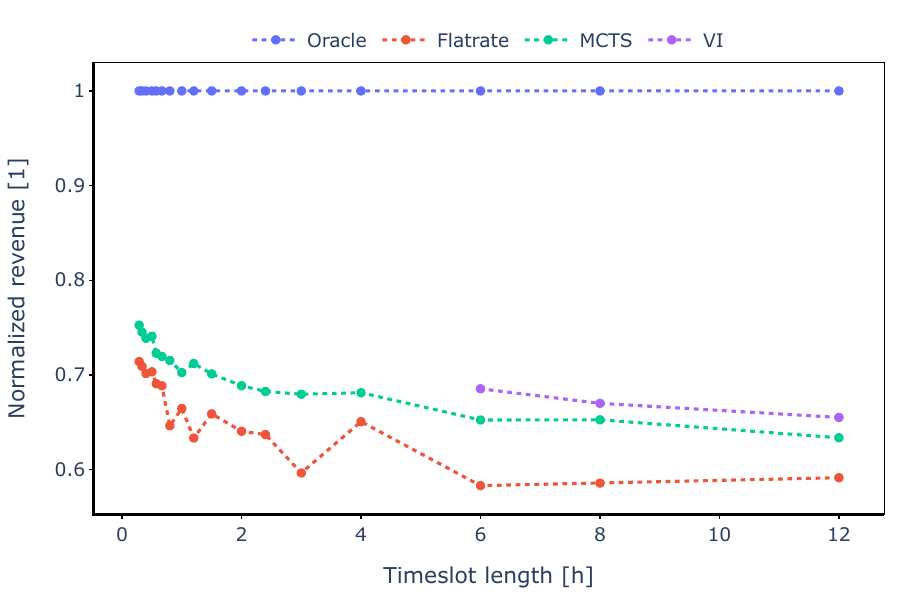}
  \caption{
      Performance of the evaluated methods when refining the discretization of the charging resources (and reducing timeslot length) in the \num{24} hour charging window with fixed expected charging demand (\num{24} hours of charging). VI runs out of memory for a timeslot shorter than 6 hours. The demand process discretization error is kept constant by increasing the number of timesteps. 
  }
  \label{fig:normalized_revenue_vs_timeslot_length}
\end{figure}

\begin{figure}
  \centering
  \includegraphics[width=0.99\textwidth]{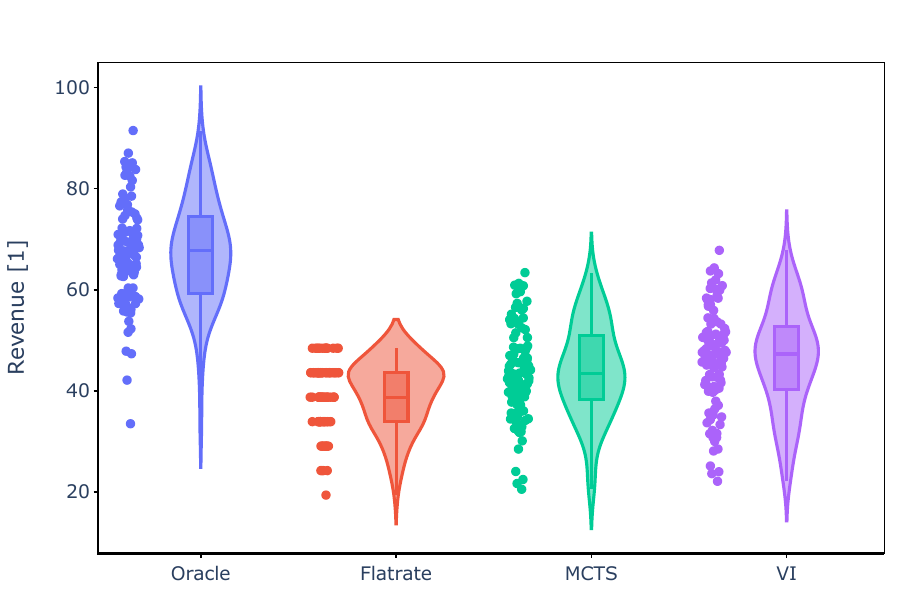}
  \caption{
    Violin plot showing the performance of the \MCTS/ algorithm compared to the flatrate, VI, and Oracle baselines in small instances. The vertical axis is the mean revenue averaged over 100 instances. The Oracle baseline provides a theoretical bound on the best achievable performance.
  } 
  \label{fig:violin_revenue}
\end{figure}

\begin{figure}
  \centering
  \includegraphics[width=0.99\textwidth]{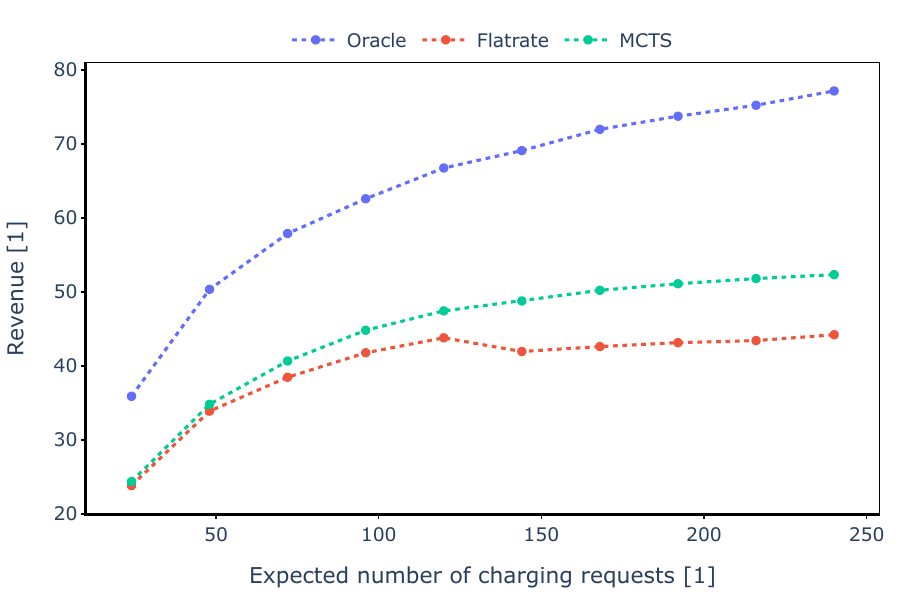}
  \caption{
    Performance of the evaluated methods when increasing demand.  Number of timeslots is fixed at \num{24}. VI does not run for these problem instances due to memory constraints.
  }
  \label{fig:revenue_vs_expected_res}
\end{figure}

\begin{figure*}

  \centering
  \def\svgwidth{0.99\textwidth}
  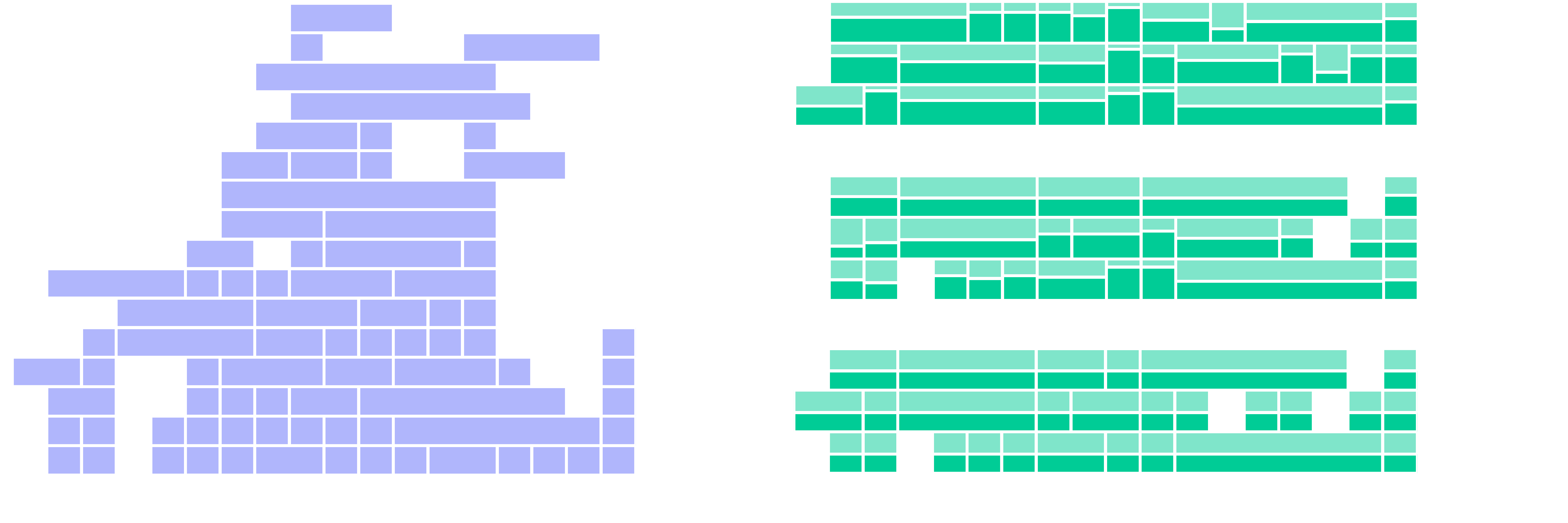
  \caption{
    Allocation of the charging requests in a single instance of the problem. The left-hand side (blue) is all the charging requests in the instance. The right-hand side shows the allocation of the requests by the different methods, with the revenue from each request illustrated by the proportion of the darker color. The horizontal axis is the time of the day; the vertical axis is the CS. The row also roughly corresponds to the order of request arrivals, that is, the later the charging request was made, the higher it is in the figure.
  } 
  \label{fig:trace_analysis}
\end{figure*}

In the next experiments, we evaluate the performance of the \MCTS/ algorithm against the flatrate, VI, and Oracle baselines in a set of instances with varying timeslot and timestep lengths to demonstrate the scalability of our approach.
The problem instance again models a CS with three chargers in a \num{24} hour day and expected demand for \num{24} hours of charging.
The number of charging timeslots is varied from \num{2} timeslots, each  \SI{12}{\hour},  to \num{240} timeslots at \SI{15}{\minute} each. 
The number of timesteps in each instance is then calculated using \cref{eq:relative_error} to make the relative discretization error constant at \num{0.06} missed charging requests per one expected request. 
Therefore, the number of timesteps varies between \num{16} timesteps at \SI{1.5}{\hour} each to \num{768} approximately \SI{2}{\minute} timesteps.

Detailed results are shown in \cref{tab:results}, the table shows averages from \num{100} instances.
For VI, we show all results; for the rest of the methods, we show results for instances with timeslots of length \SI{20}{\minute} and its multiples. 
\Cref{fig:normalized_revenue_vs_timeslot_length} shows the performance of the evaluated methods, averaged from \num{100} runs, normalized to the oracle baseline. 
The optimal VI baseline runs out of memory for more than \num{4} resources.
Performance of all methods improves with as timeslot length decreases. 
The flatrate baseline performs the worst, as expected. 
The \MCTS/ algorithm keeps close to the performance of the VI baseline and outperforms the flatrate baseline. 

For the \num{100} instances with \SI{6}{\hour} timeslots, we also give a sense of the variance of the results in \cref{fig:violin_revenue}, which shows that the \MCTS/ algorithm, stochastic by nature, does not significantly increases the variance of the results over the other baselines.

\subsubsection{Varying Expected Demand}
For the second experiment, we fix the timeslot length at \SI{1}{\hour}, and we vary the expected number of requests. 
The results of this experiment are presented in \cref{fig:revenue_vs_expected_res}. 
As can be expected, the revenue of all the pricing methods increases with increasing demand, and the gap between \MCTS/ and flatrate grows with demand. 
The VI baseline does not run in these instances. 

The number of expected requests directly influences the error caused by the discretization of the demand process (see \cref{eq:relative_error}). Therefore, to control the error, we have to vary the number of timesteps in this experiment. We calculate the number of timesteps and their length to keep the relative error constant again at \num{0.06} ignored charging requests per one expected request. 
This results in \num{192} timesteps at \SI{7.5}{\minute} each for the instances with the smallest demand, and \num{1920} timesteps, each \SI{45}{\second}, for the instances with the highest number of expected requests.

\subsubsection{Focusing on Single Instance}

The mechanism by which the \MCTS/ and Oracle baseline outperform the flatrate is illustrated in \cref{fig:trace_analysis}. The figure shows the allocation of the charging requests in a single instance of the problem with a timeslot length of one hour. In this high-demand instance, all pricing strategies allocate almost all of the available capacity (3 CSs). However, the \MCTS/ algorithm and the Oracle baseline manage to extract more revenue either from the same requests or from requests that the flatrate baseline does not allocate as it runs out of capacity earlier. 

\section{Conclusions}
In conclusion, we have described an \MDP/ based model for dynamic pricing of \emph{reserved EV charging  as complete service}, including charging, reservation fee, parking, etc. This type of charging service is not well-studied in the literature, but could be an important part of the future EV charging infrastructure as more EVs get on roads. 
By pricing the service reservations dynamically, the CS operator can maximize his revenue while providing a guaranteed charging service to EV drivers who plan their trips in advance. 

Our pricing relies customer demand modelled as a Poisson process. 
For use in an \MDP/ we need to discretize the process into timesteps. 
We show that the discretization introduces error, and as a novel contribution, we quantify this error and show how it can be controlled by increasing the number of timesteps.
In experiments, we show how we can control the tradeoff between this approximation error and computational complexity.

To find a dynamic pricing policy from \MDP/ model, we propose a novel \MCTS/ heuristic. 
We show that this pricing algorithm outperforms the flatrate baseline, and is competitive with the optimal VI baseline in instances where this baseline is available. 
We also show that this \MCTS/ algorithm is scalable to larger instances and that its performance does not deteriorate with regard to the globally optimal Oracle baseline.

We believe that both the model and the approximation error term we introduced in this work provide a practical tool for future researchers working on similar models.  
For the future work, we see two main directions. 
First, we would like to evaluate the performance of the pricing algorithm on the real-world data, if appropriate training data becomes available.
This will allow us to verify the impact of the discretization error on the performance of the pricing algorithm. 
Second, we would like to validate the approach of dynamic pricing of the reserved charging services in combination with the multi-criterial optimization of the driver trips~\cite{cuchyMultiObjectiveElectricVehicle2024}. 
Joining the two techniques in a multi-agent simulation, such as in~\cite{basmadjianReferenceArchitectureInteroperable2020}, would allow us to evaluate the impact of our pricing scheme on the drivers as well as the CS operator.

\section{Acknowledgments}
This work was supported by the project of CTU no. SGS/22/168/OHK3/3T/13, ``Decision-making in large, dynamic multi-agent scenarios 2''and the Czech Science Foundation (grant no. GA25-18353S).



\printbibliography

\end{document}